\documentclass[sigconf]{acmart}  % ,review
\usepackage{amsmath,amsfonts}  

\usepackage{graphicx}
\usepackage{textcomp}
\usepackage{xcolor}
\usepackage{color}
\usepackage{setspace}
\usepackage{hyperref}
\usepackage{subfigure}
\usepackage{algorithm}
\usepackage{algpseudocode} 
\usepackage{tablefootnote}
\usepackage{multirow}
\usepackage{booktabs}
\usepackage{subfigure}
\usepackage{fancyhdr}
\usepackage{pifont}
\newcommand{\blue}[1]{\textcolor{black}{#1}}

\AtBeginDocument{%
  \providecommand\BibTeX{{%
    \normalfont B\kern-0.5em{\scshape i\kern-0.25em b}\kern-0.8em\TeX}}}

\setcopyright{acmcopyright}
\copyrightyear{2018}
\acmYear{2018}
\acmDOI{XXXXXXX.XXXXXXX}
\acmConference[WWW '23]{WWW '23: ACM Web Conference}{April 30--May 4, 2023}{Austin, TX, US}
\acmBooktitle{Proceedings of the ACM Web Conference 2023
(WWW ’22), April 30--May 4, 2023, Austin, TX, US}
\acmPrice{15.00}

\begin{document}

%%
%% The "title" command has an optional parameter,
%% allowing the author to define a "short title" to be used in page headers.
\title[RiskProp]{RiskProp: Account Risk Rating on Ethereum via De-anonymous Score and Network Propagation}%Please review title change.

%%
%% The "author" command and its associated commands are used to define
%% the authors and their affiliations.
%% Of note is the shared affiliation of the first two authors, and the
%% "authornote" and "authornotemark" commands
%% used to denote shared contribution to the research.

%\author{Anonymous author(s)}

\author{Dan Lin}
%\authornote{Both authors contributed equally to this research.}
\orcid{0000-0001-7067-2396}
%\author{G.K.M. Tobin}
%\email{webmaster@marysville-ohio.com}
\affiliation{%
  \institution{School of Software Engineering, \\Sun Yat-sen University}
  \city{Zhuhai}
  \country{China}}
\email{lind8@mail2.sysu.edu.cn}

\author{Jiajing Wu}
\authornote{Corresponding author.}
\affiliation{%
  \institution{School of Computer Science and \\Engineering, Sun Yat-sen University}
  \city{Guangzhou}
  \country{China}}
\email{wujiajing@mail.sysu.edu.cn}

\author{Qishuang Fu}
\affiliation{%
  \institution{School of Computer Science and \\Engineering, Sun Yat-sen University}
  \city{Guangzhou}
  \country{China}}
\email{fuqsh6@mail2.sysu.edu.cn}

\author{Zibin Zheng}
\affiliation{%
  \institution{School of Software Engineering, \\Sun Yat-sen University}
  \city{Zhuhai}
  \country{China}}
\email{zhzibin@mail.sysu.edu.cn}

\author{Ting Chen}
\affiliation{%
  \institution{University of Electronic Science and
Technology of China}
  \city{Guangzhou}
  \country{China}}
\email{brokendragon@uestc.edu.cn}

%%
%% By default, the full list of authors will be used in the page
%% headers. Often, this list is too long, and will overlap
%% other information printed in the page headers. This command allows
%% the author to define a more concise list
%% of authors' names for this purpose.
\renewcommand{\shortauthors}{Anonymous author(s)}

% 删除多余的reference信息
% https://www.twblogs.net/a/5f020cc7d496dddbb5425051/?lang=zh-cn
\settopmatter{printacmref=false} 
\renewcommand\footnotetextcopyrightpermission[1]{}

%%
%% The abstract is a short summary of the work to be presented in the
%% article.
\begin{abstract}
  As one of the most popular blockchain platforms supporting smart contracts, Ethereum has caught the interest of both investors and criminals. Differently from traditional financial scenarios, executing Know Your Customer verification on Ethereum is rather difficult due to the pseudonymous nature of the  blockchain. Fortunately, as the transaction records stored in the Ethereum blockchain are publicly accessible, we can understand the behavior of accounts or detect illicit activities via transaction mining. Existing risk control techniques have primarily been developed from the perspectives of de-anonymizing address clustering and illicit account classification. However, these techniques cannot be used to ascertain the potential risks for all accounts and are limited by specific heuristic strategies or insufficient label information. These constraints motivate us to seek an effective rating method for quantifying the spread of risk in a transaction network. To the best of our knowledge, we are the first to address the problem of account risk rating on Ethereum by proposing a novel model called RiskProp, which includes a de-anonymous score to measure transaction anonymity and a network propagation mechanism to formulate the relationships between accounts and transactions. We demonstrate the effectiveness of RiskProp in overcoming the limitations of existing models by conducting experiments on real-world datasets from Ethereum. Through case studies on the detected high-risk accounts, we demonstrate that the risk assessment by RiskProp can be used to provide warnings for investors and protect them from possible financial losses, and the superior performance of risk score-based account classification experiments further verifies the effectiveness of our rating method.
\end{abstract}


%%
%% This command processes the author and affiliation and title
%% information and builds the first part of the formatted document.
\maketitle

\section{Introduction}

% P1: 介绍被研究的对象，引出研究问题

Ethereum~\cite{wood2014ethereum} has the second-largest market cap in the blockchain ecosystem. The account model is adopted on Ethereum, and the native cryptocurrency on Ethereum is named Ether (abbreviated as ``ETH''), which is widely accepted as payments and transferred from one account to another. It is known that Ethereum accounts are indexed according to pseudonyms, and the creation of accounts is almost cost-free. This anonymous nature and the lack of regulation result in the bad reputation of Ethereum and other blockchain systems %in
for breeding malicious behaviors and enabling fraud, thereby resulting in large property losses for investors. As reported in a \href{https://go.chainalysis.com/2021-crypto-crime-report.html}{\textit{Chainalysis Crime Report}}, the illicit share of all cryptocurrency activities was valued at nearly USD %Please check intended meaning is retained.
2.7 billion in 2020.
These losses illustrate that Know-Your-Customer (KYC) and risk control of accounts are critical and necessary. Risk control~\cite{risk2014Malte} not only helps wallet customers identify risky accounts and avoid losses but also plays a vital role in the anti-money laundering of virtual asset service providers, such as cryptocurrency exchanges.

Therefore,  a wealth of efforts have been expended in risk control on Ethereum in recent years. In September 2020, the Financial Action Task Force (FATF) published a recommendation report %of
on virtual assets and released information on \textit{Red Flag Indicators}~\cite{fatf2022} related to transactions, anonymity, senders or recipients, the source of funds, and geographical risks. 
%Besides, the academic community has proposed various techniques from the perspectives of address clustering and illicit account classification.
In addition, researchers in the academic community have proposed various techniques from the perspectives of address clustering and illicit account classification.
Address clustering techniques perform entity identification of anonymous accounts. For example, Victor~\cite{victor2020address} proposes several clustering heuristics for Ethereum accounts and clusters 17.9\% of all active externally owned accounts. 
Illicit account detection techniques focus on training classifiers based on well-designed features extracted from transactions~\cite{Farrugia2020,trans2vec2020Wu,Chen2020Phishing}. Moreover, some researchers have developed methods for automatic feature extraction incorporating structural information~\cite{elliptic2019Weber, yuan2020phishing, modeling2020lin,Li2022TTAGN}. 

% Recently, companies has developed anti-money laundering platforms and provided query services of malicious wallet addresses based on various data sources, e.g. CoinHolmes transaction path visualization~\footnote{\url{https://trace.coinholmes.com/}} of PeckShield and SlowMist threat intelligence engine~\footnote{\url{https://aml.slowmist.com/en/maliciousWallet.html}}.   
% The international organizations also pay attention to this issue. In September 2020, Financial Action Task Force (FATF) published a recommendation report of Virtual Asset and released Red Flag Indicators~\citep{fatf2022} related to transactions, anonymity, senders or recipients, the source of funds and the geographical risks. 
% At the same time, the academic community has presented account classification models on Ethereum~\citep{Wu2021JNCA}. Existing methods can be divided into two categories. One relies on the handcraft feature based on the transaction records, and builds supervised classifiers with small amount of labels~\citep{Farrugia2020,trans2vec2020Wu,Chen2020Phishing}. The other one notices the interactions between accounts and takes more steps into automatic feature extraction. They incorporate the graph structural information by graph embedding techniques, like Graph Convolutional Network (GCN)~\cite{elliptic2019Weber,yuan2020phishing,chen2021tegdetector}.

% P3：现有视角存在的局限性, 引出什么是Account Risk Rating，为什么它能解决已有视角的局限性

However, there are still some limitations (L) associated with these techniques. 
%qs comment:and most accounts are  beyond heuristic rules cannot thus be identified. 删除了are
\textbf{L1:} Account clustering techniques 
can only be applied to part of accounts and therefore have limited applicability, and most accounts  beyond heuristic rules cannot thus be identified. %Please check intended meaning is retained.
\textbf{L2:} In the existing methods for illicit account detection, \textit{binary} classifiers are usually trained via \textit{supervised} learning. %but few accounts are labeled. 
However, as only a very small percentage of risky nodes have clear labels, which are required for these methods, the vast majority of accounts that may be involved in malicious events are unlabeled. In particular, risky accounts with few transactions or unseen patterns are likely to be misidentified in practical use. 
%L3: Also, the risk propagation through money flow is ignored in the feature extraction procedures. 

To address the limitations presented above, we explore  risk control on Ethereum from a new perspective: \textit{Account risk rating}. %\textcolor{red}{In traditional financial scenarios, credit scoring are usually conducted by authorized financial institutions,} 
In traditional financial scenarios, credit scoring is usually conducted by authorized financial institutions, which perform audits on their customers to fully understand their identity, background,  and financial credit standing. Similarly to credit scoring, risk rating on Ethereum can help us quantify the 
latent risk of a transaction or account with a quantitative score, thereby combating money laundering and identifying potential scams before %before becoming the next victim. 
% the fraudsters cheat new victims. 
new victims emerge.
In terms of the abovementioned \textbf{L1}, in contrast to the traditional account clustering method, which can only de-anonymize a small number of accounts, the account risk method proposed in this paper can obtain quantitative risk indicators for \textit{all} accounts.
%\textcolor{red}{While for \textbf{L2}, the proposed risk rating method can achieve decent performance in a unsupervised manner without feeding labels,} 
%qs comment: and the output is risk values改成了新的一句话，这样就避免which指代不清了。
Regarding \textbf{L2}, the proposed risk rating method can achieve decent performance in an unsupervised manner without feeding labels. The output of  the proposed method is risk values, which are provided continuously and allow evaluation of the severity of risk.
% (rather than binary results). 

% \subsection{Transaction Network Construction}
\begin{figure}%[htbp]
  \centering
  \includegraphics[width=0.8\linewidth]{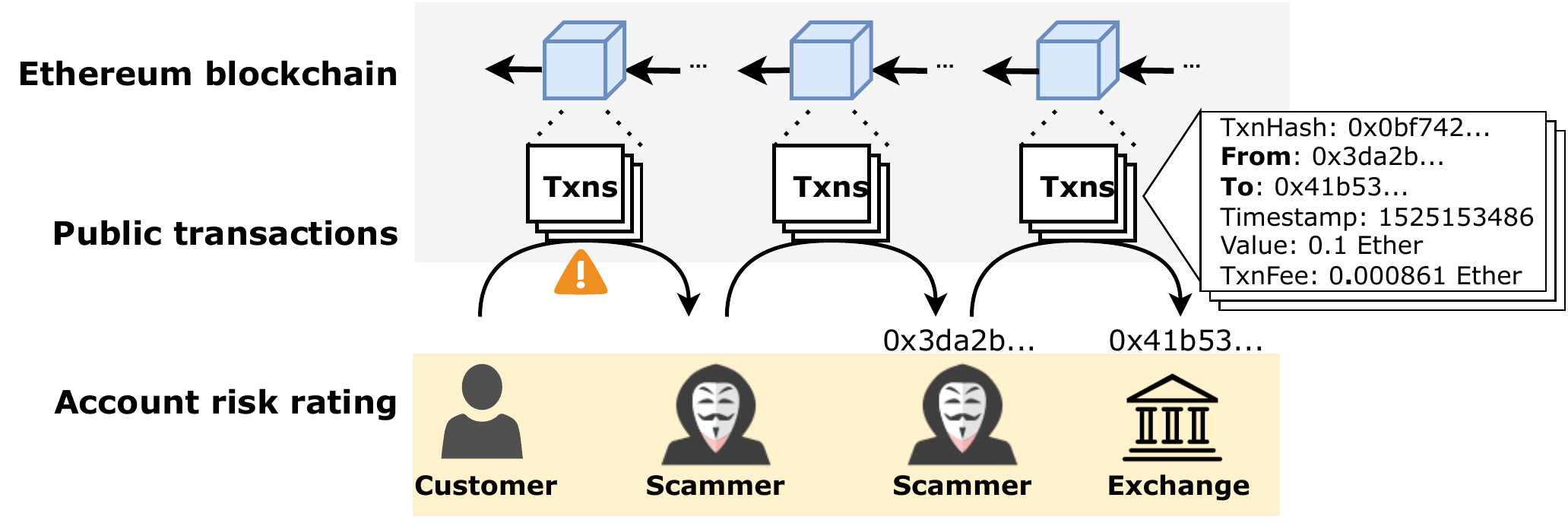}
  \vskip -2ex
  \caption{The procedure of ETH transfer in Ethereum. ``From'' denotes the sender, ``To'' denotes the receiver, and ``Txn'' denotes ``Transaction''.}
  \vskip -2ex
  \label{fig:background}
\end{figure}

Compared with traditional financial scenarios, several unique challenges (C) are encountered in the task of account risk rating on Ethereum. %Please check intended meaning is retained.
\textbf{C1: Nature of anonymity.} Transactions on Ethereum do not require real-name verification. Even worse, perpetrators of some malicious activities deliberately enhance their anonymity to counter the impact of de-anonymizing clustering techniques~\cite{elliptic2019Weber}. 
\textbf{C2: Complex transaction relationship.}
Compared with traditional financial scenarios, a user or entity on Ethereum may control a large number of accounts at almost no cost, and the transaction relationship between accounts is also more complex. 
How to quantify the impact of trading behavior between accounts on account risk is a challenging core problem.
%Thus it is very difficult and inefficient to investigate the identity and background of Ethereum accounts. 

To overcome the challenges mentioned above, we propose a novel approach called Risk Propagation (\textit{RiskProp}) for Ethereum account rating. It comprises two core designs, namely de-anonymous score and a network propagation mechanism.
To resolve \textbf{C1}, de-anonymous score measures the degree %of
to which transactions remain anonymous. For example, both the payer and the payee of an illicit transaction prefer to have a small number of transactions to ensure anonymity-preserving protection. In contrast, both sides of a licit transaction may participate in numerous interactions without evading the impact of the de-anonymized clustering algorithm. 
Afterward, to resolve \textbf{C2}, we model the massive transaction records as a directed bipartite graph and introduce a network propagation mechanism with three interdependent metrics, namely \textit{Confidence} of the de-anonymous score, \textit{Trustiness} of the payee, and \textit{Reliability} of the payer. Intuitively, payees with higher trustiness receive transactions with higher de-anonymous scores, and payers with higher reliability will send transactions with higher confidence. Clearly, reliability, trustiness, and confidence are related to each other, so we define five items of prior knowledge that these metrics should satisfy and propose three mutually recursive equations to estimate the values of these metrics. 
To verify the effectiveness of the proposed risk rating method and further illustrate the significance of rating for risk control on Ethereum, we evaluate the effect of the risk rating system via experiments from two aspects, i.e., analysis of risk rating results and rating score-based illicit/licit classification.

Overall, our contributions are summarized as follows:

\noindent $\bullet$ \textbf{A new perspective for Ethereum risk control. } This paper is the first to propose tackling the problem of Ethereum risk control via the perspective of account risk rating.
    
\noindent $\bullet$ \textbf{A novel risk metric for transactions. } We creatively develop a metric called \textit{de-anonymous score} for transactions, which measures the degree of de-anonymization to quantify the risk of a transaction.
    
\noindent $\bullet$ \textbf{An effective method and interesting insights.} We implement a novel risk rating method called \textit{RiskProp} and demonstrate its superior effectiveness and efficiency via experiments on a real-world Ethereum transaction dataset together with theoretical analysis. By analyzing the rating results and case studies on high-risk accounts, we obtain interesting insights into the Ethereum ecosystem and further show how our method could prevent financial losses ahead of %Please check intended meaning is retained.
blacklisting malicious accounts.
    %\item \textbf{Theoretical guarantee:} \textit{RiskProp} is guaranteed to converge and has linear time complexity (See Section~\ref{sec:Supplement}). % within a number of interactions, 

% Our contributions are as followed:
% \begin{itemize}
%     \item We are the first attempt on Ethereum account risk assessment task that quantify the risk propagation with graph topology for better modeling account intentions. We propose to estimate the risk value of accounts according to public transactions recorded on Ethereum blockchain. % and regulation requirement 
%     \item We design and implement a novel graph-based framework called \textit{RiskProp} to solve the risk assessment for cryptocurrency accounts. \textit{RiskProp} is open-sourced and reproducible~\footnote{See code in \url{https://github.com/lindan113/RiskProp}}.
%     \item Experiments on real-world dataset of Ethereum blockchain demonstrate the effectiveness of RiskProp in both unsupervised and unsupervised evaluations. Furthermore, our model provides case studies on predicting accounts beyond the existing label library.
% \end{itemize}

% P4：新的研究视角有什么研究挑战，以及我们如何解决

% \begin{figure}%[htbp]
%   \centering
%   \includegraphics[width=\linewidth]{background.pdf}
%   \vskip -1.5ex
%   \caption{The procedure of ETH transfer in the Ethereum. ``From'' denotes the sender, ``To'' denotes the receiver, and ``Txn'' denotes ``Transaction''.}
%   \vskip -1.5ex
%   \label{fig:background}
% \end{figure}

\vspace{-1.5ex}
\section{Preliminary}

\subsection{Ethereum Financial Background}

Ether is the native ``currency'' on Ethereum and plays a fundamental part in the Ethereum payment system. Ether can be paid or received in financial activities, just like currency in real life. %\textcolor{red}{In conventional financial scenarios, know your customer (KYC) or KYC check is the mandatory process to identify and verify the customer's identity when opening an account,} 
In conventional financial scenarios, a Know Your Customer (KYC) check is the mandatory process to identify and verify a customer's identity when opening an account and to periodically understand the legitimacy of the involved funds over time. 
However, unlike traditional transaction systems, where customers' identity information is required and obtained in KYC checks, Ethereum accounts are designed as pseudonymous addresses identified by 20 bytes of public key information generated by cryptographic algorithms, for example,
``0x99f154f6a393b088a7041f1f5d0a7cbfa795d301''. 
% and similarly, investors can keep their Ether in their crypto-wallets. 

% \begin{figure}%[htbp]
%   \centering
%   \includegraphics[width=\linewidth]{background.pdf}
%   \vskip -1.5ex
%   \caption{The procedure of ETH transfer in the Ethereum. ``From'' denotes the sender, ``To'' denotes the receiver, and ``Txn'' denotes ``Transaction''.}
%   \vskip -1.5ex
%   \label{fig:background}
% \end{figure}

% \begin{figure}%[htbp]
%   \centering
%   \includegraphics[width=0.8\linewidth]{Fig_background.pdf}
%   \vskip -1.5ex
%   \caption{The procedure of ETH transfer in Ethereum. ``From'' denotes the sender, ``To'' denotes the receiver, and ``Txn'' denotes ``Transaction''.}
%   \vskip -1.5ex
%   \label{fig:background}
% \end{figure}

Figure~\ref{fig:background} depicts the risky scenario of Ether transfer in aspects of data acquisition. It includes three layers: 1) Ethereum blockchain. The Ethereum historical data are irreversible and publicly traceable on the chain. 2) Public transactions. The transaction denotes a signed data package from an account to another account, including the sending address, receiver address, transferred Ether amount, etc. 3) Account risk rating. Usually, the identities who control the accounts are not labeled. Customers may become involved in suspicious financial crimes or be vulnerable to frauds and scams. Furthermore, the illicit funds can be laundered and cashed out via exchanges. In this procedure, our proposed \textit{RiskProp} is implemented to measure the risk of unlabeled accounts that may have ill intentions and alert customers when engaging in suspicious, potentially illegal transactions.

%``0x99f154f6a393b088a7041f1f5d0a7cbfa795d301''.
\vspace{-1.5ex}
\subsection{The Nature of Blockchain: Anonymity}

It is known that the Ethereum account is identified as a pseudonymous address. %\textcolor{red}{ However, if customers repeatedly use the same address as on-chain identification, the relationship between accounts are linkable by the public transaction records. } 
However, if customers repeatedly use the same address as on-chain identification,  the relationship between addresses becomes linkable via public transaction records. 
Accounts that participate in more transactions and connect with more accounts experience degrading anonymity~\cite{elliptic2019Weber}. %\textcolor{red}{To reduce the likelihood of exposure, criminals naturally tend to initiate transactions with less accounts.} 
To reduce the likelihood of exposure, criminals naturally tend to initiate transactions with fewer accounts.
Here is an example on Ethereum: The two accounts of transaction \href{https://cn.etherscan.com/tx/0x9a9df5e68a8cbb30e901024367c88dd3401b455a5916204753ebda6070bccde9}{0x9a9d} have only three transactions and became inactive thereafter. %Please check intended meaning is retained.
These two accounts are considered suspicious and reported as relevant \href{https://cn.etherscan.com/accounts/label/upbit-hack}{accounts of Upbit exchange hack}.
%\footnote{\url{https://medium.com/merkle-science/hack-track-upbit-cryptocurrency-exchange-b1f17baa5a72}}
On the contrary, entities who do not deliberately take anonymity-preserving measures are likely to be normal~\cite{elliptic2019Weber}. Thus, the transaction is scored based on the fact of whether the accounts are trying to hide or not, which is the \textit{de-anonymous score}.

%Based on the nature of Ethereum, the quality Score for cryptocurrency transactions is defined as \textit{de-anonymous Score}.
\textbf{Definition 1 (De-anonymous score, abbreviated as ``score'').}
The de-anonymous score of a transaction from account $u$ to $v$ where there is no intention to hide is defined as

\begin{equation}
    \begin{split}
    Score(u, v) =  & \frac{1}{2}(\frac{2\log |OutTxn(u)|-\log maxOut}{\log maxOut} \\ 
    &+ \frac{2\log |InTxn(v)|-\log maxIn}{\log maxIn}),
    \end{split}
    \label{equ:Score}
\end{equation}
% \begin{equation}
%     \begin{split}
%     Score(u, v) =  & \frac{1}{2}(\frac{2\log |Out(u)|-\log maxOut}{\log maxOut} \\ 
%     &+ \frac{2\log |In(v)|-\log maxIn}{\log maxIn}),
%     \end{split}
%     \label{equ:Score}
% \end{equation}
%增加说明out(u) 和in(v)的最小值为1
% where $|Out(u)|$ represents the number of payments of payer $u$, and $|In(v)|$ represents the number of receptions of payee $v$. The minimum value of $|Out(u)|$ and $|In(v)|$ is 1.
%qs comment:and $|\times |$ denotes the size of set在set前面加了a。
where $OutTxn(u)$ represents the outgoing transactions (payments) of payer $u$, $InTxn(v)$ represents the incoming transactions (receptions) of payee $v$, and $|\times |$ denotes the size of a set. The minimum value of $|OutTxn(u)|$ and $|InTxn(v)|$ is 1.
Let $maxOut$ and $maxIn$ be the largest number of payments and receptions, respectively.  
The de-anonymous scores of a transaction $(u, v)$ range from $-1$ (very high anonymity, abnormal) to 1 (very low anonymity, normal). Intuitively, the score of $(u, v)$ increases as the transaction numbers of either payer or payee grow.
Note that tricky criminals may camouflage themselves by deliberately conducting low-anonymity transactions ~\cite{flowscope2020}.

\subsection{Transaction Network Construction}
First, each transaction on Ethereum has one payer (i.e., sender) and one payee (i.e., receiver). 
Any account can be the role of payer or payee, just as a person in real life has different roles. 
The payee is a \textit{passive} role and, therefore, we consider the incoming transactions to indicate the trustiness of an account. For instance, exchange accounts that receive more transactions are considered to be more trustworthy. 
In contrast, the payer is an \textit{active} role and, thus, the outgoing transactions embody the intention of an account. For example, a scam account subjectively wants to transfer stolen money to its partners.

% \begin{figure}[t]
%   \centering
%   \includegraphics[width=0.8\linewidth]{transformation-0607.pdf}
%   \vskip -1.5ex
%   \caption{The transformation from the raw transaction records to the directed bipartite graph. ``Txn'' denotes ``Transaction''.}
%   \vskip -1.5ex
%   \label{fig:networkconstruction}

% \end{figure}

Next, the transaction records are modeled as a directed bipartite graph $G=(U, V, S)$, where $U$, $V$, and $S$ represent the set of all payers, payees, and scores, respectively. A weighted edge $(u, v)$ denotes the transfer of Ethers from account $u\in U$ to account $v\in V$ with $Score(u, v) \in S$. The graph construction procedure is shown in Figure~\ref{fig:networkconstruction}.

Then, the ego network of a payer $u$ is introduced. It is formed by its outgoing scores and corresponding payee neighbors, formulated as $Out(u)\cup \{v|(u,v)\in Out(u)\}$, where $Out(u)$ is the set of scores connected with $u$. It is similar for the ego network of a payee, formulated as $In(v)\cup \{u|(u,v)\in In(v)\}$. Figure~\ref{fig:example} shows an example in which there are two payers, two payees, and three transactions.

%Two payer $u_1$ and $u_2$ are said to have identical egonetworks if $|Out(u_1)|=|Out(u_2)|$ , and there exists a one-to-one mapping $h:Out(u_1)$

\begin{figure}%[htbp]
  \centering
    \includegraphics[width=0.9\linewidth]{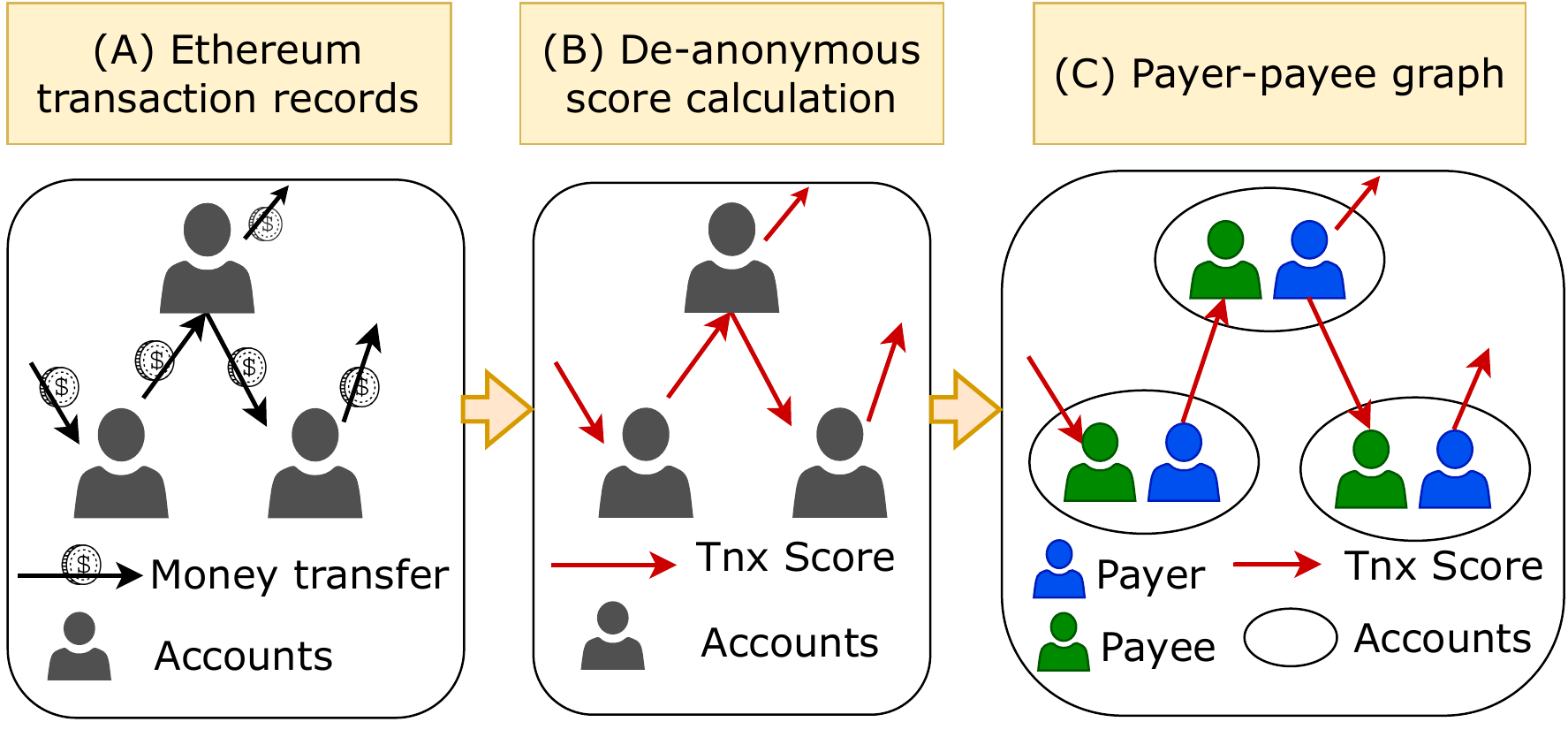}
  \vskip -1.5ex
  \caption{The transformation from the raw transaction records to the directed bipartite graph. ``Txn'' denotes ``Transaction''.}
  \vskip -1.5ex
  \label{fig:networkconstruction}

\end{figure}

\begin{figure}[t]
  \centering
  \includegraphics[width=0.95\linewidth]{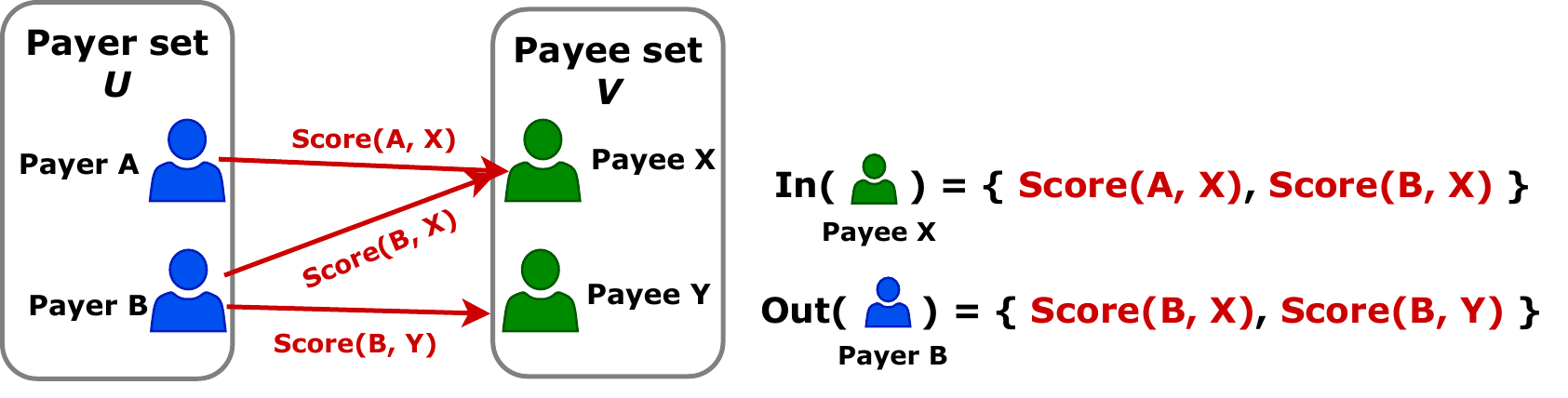}
  \vskip -1.5ex
  \caption{A toy example of the directed bipartite graph established from transactions and the illustration of functions $In$ and $Out$.}
  
  \vspace{-1.5ex}
  
  \label{fig:example}

\end{figure}

\section{Model}

In this section, we describe the prior knowledge that establishes the relationships among accounts and transactions and then propose risk propagation formulations that satisfy the prior knowledge. It is worth noticing that the proposed algorithm does not require handcraft feature engineering.

\subsection{Problem Definition and Model Overview}

Given raw transaction records of Ethereum, we model the transaction relationships between accounts as a directed bipartite graph $G=(U, V, S)$ with payers and payees as nodes and prepossessed de-anonymous scores as weights of edges.
We believe that accounts have intrinsic metrics to quantify their reliability and trustworthiness and transactions have intrinsic metrics to measure the confidence of their calculated de-anonymous scores.
Naturally, those metrics are interdependent and interplay with each other via the risk propagation mechanism:

%qs comment:把 itself改成himself
    \noindent $\bullet$ \textbf{Payers} vary in terms of their \textit{Reliability}, which indicates how motivated they are. A licit payer without malicious intent usually does not hide himself or disguise its intentions during transactions. Specifically, a reliable payer has harmless intentions regardless of whether it is transferring money to an exchange or to a scammer account (being gypped). In contrast, a perpetrator (e.g., a scammer) hopes to cover up its traces~\cite{victor2020address}. The reliability metric $R(u)$ of a payer $u$ lies in $[0,1]$, $\forall u\in U$. A value of 1 denotes a 100\% reliable payer and 0 denotes a 0\% reliable payer.
    
    \noindent $\bullet$ \textbf{Payees} vary in their %trustworthy
    trustworthiness level, measured by a metric called \textit{Trustiness}, which indicates how trustworthy they are. Intuitively, a cryptocurrency service provider with a better reputation will receive more licit transactions (with higher scores) from well-motivated payers. Trustiness of a payee $T(v)$ ranges from 0 (very untrustworthy) to 1 (very trustworthy) $\forall v\in V$.
    
    \noindent $\bullet$ \textbf{De-anonymous scores} vary in terms of \textit{Confidence}, which reflects the confidence in the estimated risk probability of a transaction. The confidence metric $Conf(u,v)$ ranges from 0 (lack of confidence) to 1 (very confident).

\textit{The connection between the reliability and risk of accounts:}
We define \textit{Reliability} to characterize the risk rating of accounts because an account's intention can be inferred by its (active) sending behavior, rather than by its (passive) receiving behavior. A scammer transferring stolen money to its gang is a better reflection of its evil intention than the receipt of stolen money from victims. In the later section, we calculate the \textit{risk} rating of accounts based on the \textit{Reliability} of payer roles. %Please check intended meaning is retained.

% money launder transfer stolen funds to their gang accounts, and the Ponzi investment scammers may send money back to the investors, to create the illusion of profitability.

\subsection{Network Propagation Mechanism}
Given a cryptocurrency payer--payee graph, all intrinsic metrics are unknown but are interdependent. Here, we introduce five items of prior knowledge that establish the relationships and how the network propagation mechanism is specially designed for our problem. The first two items of prior knowledge reflect the interdependency between a payee and the de-anonymous scores that they receive.

\noindent \textbf{[Prior knowledge 1] }
Payees with higher trustiness receive transactions with higher de-anonymous scores.
Intuitively, a payee receiving transactions with high de-anonymous scores is more likely to be trustworthy. Formally, if two payees $v_1$ and $v_2$ have a one-to-one mapping,  $h:In(v_1)\rightarrow In(v_2)$ and $Score(u,v_1) > Score(h(u), v_2)$ $\forall(u,v_1)\in In(v_1)$, then $T(v_1)>T(v_2)$.

\noindent \textbf{[Prior knowledge 2]}
Payees with higher trustiness receive transactions with more positive confident scores. 
For two payees $v_1$ and $v_2$ with identical de-anonymous score networks, if the confidence of the in-transactions of payee $v_1$ is higher than that of payee $v_2$, the trustiness of payee $v_1$ should be higher. Formally, if two payees $v_1$ and $v_2$ have a one-to-one mapping, $h:In(v_1)\rightarrow In(v_2)$ and $Conf(u,v_1) > Conf(h(u), v_2)$ $\forall(u,v_1)\in In(v_1)$, then $T(v_1)>T(v_2)$.

According to the above prior knowledge, we develop the \textit{Trustiness} formulation for $\forall v \in V$ of our \textit{RiskProp} algorithm:
\begin{equation}
    T(v) = \frac{\sum_{(u,v)\in In(v)} Score(u,v)\times Conf(u,v)}{|In(v)|}.
    \label{T}
\end{equation}

The next item of prior knowledge defines the relationship between the score of a transaction and the connected payer--payee pair using the %anonymity
anonymous nature of cryptocurrency.

% \begin{figure*}[t]
%   \centering

%   \includegraphics[width=0.7\linewidth]{workflow-update.pdf}
%   \caption{The workflow of account risk rating on Ethereum. }
%   \label{fig:workflow}
% \end{figure*}

%比较简洁的版本
\begin{algorithm}[t]
\caption{\textit{RiskProp} Algorithm}
\label{alg:riskprop}
\small
\begin{algorithmic}[1]
  \State \textbf{Input:} Directed Bipartite Graph  $G=(U,V,S)$
  \State \textbf{Output:} \emph{Risk} of accounts
  \State Initialize $T^0=0.5, R^0=0.7, Conf^0=0.5, t=0, \Delta=1$
  \While{$\Delta \geq 0.01$}
      \State $t=t+1$
      \State Update $trustiness$ of payees using Equation 2
      \State Update $reliablity$ of payers using Equation 4
      \State Update $confidence$ of transactions using Equation 3
      \State $\Delta_T= \sum_{v\in V}|T^t(v)-T^{t-1}(v)|$ 
      \State $\Delta_R=\sum_{u\in U}|R^t(u)-R^{t-1}(u)|$
      \State $\Delta_C= \sum_{(u,v)\in S}|Conf^t(u,v)-Conf^{t-1}(u,v)|$
      \State $\Delta=\max\{ \Delta_T, \Delta_R, \Delta_C\}$
  \EndWhile
  \State $Risk(u) = (1-R(u))\times 10, \forall u \in U$
  \State \Return 
\end{algorithmic}
\end{algorithm}

\begin{table}
  \centering
  \caption{An example of propagation. $R^{0}$ is initial value, $R^{final}$ and Risk$^{final}$ are the results after convergence.}
  \vskip -1.5ex
  \label{tab:example}
  \scalebox{0.8}{
  \begin{tabular}{c|cccc}
    \toprule
    Account & Label & $R^{0}$ & $R^{final}$ & Risk$^{final}$\\
    \midrule
    \midrule
    \href{https://cn.etherscan.com/address/0xa768cc13d1ab64283882ffa74255bb0564a7592b}{0xa768} & Contract-deployer & 0.7 & 0.8575 & 1.425 \\
    \href{https://cn.etherscan.com/address/0x8271b2e8cbe29396e9563229030c89679b9470db}{0x8271} & Exchange & 0.7 & 0.9526 & 0.474  \\
    \href{https://cn.etherscan.com/address/0xebdc8d0445e111aad61006f13e7a2d071a473779}{0xebdc} & Phish-hack & 0.7 & 0.1195 & 8.805  \\
    \href{https://cn.etherscan.com/address/0xfe34aa1b07c3c61305df1be55719f0be44a23cd9}0xfe34 & Phish-hack & 0.7 & 0.2330 & 7.670 \\
  \bottomrule
\end{tabular}
}
\vskip -1.5ex
\end{table}

\noindent \textbf{[Prior knowledge 3]} 
Confident de-anonymous scores of transactions are closely linked with the connected payee's trustiness. %Please check intended meaning is retained.
Formally, if two scores $(u_1, v_1)$ and $(u_2, v_2)$ are such that 
$Score(u_1, v_1)=Score(u_2, v_2)$, $R(u_1)=R(u_2)$, and $|Score(u_1,v_1)$ $- T(v_1)|\leqslant$ 
$|Score(u_2,v_2)$ $-T(v_2)|$, then $Conf(u_1, v_1)\geqslant Conf(u_2, v_2)$.

We imply that different transactions sent by the same payers can have different intentions and anonymity. Even scammers on Ethereum can have transactions that seem normal.%to camouflage} Please check intended meaning is retained.

\noindent \textbf{[Prior knowledge 4] }
Transactions with higher confidence de-anonymous scores are sent by more reliable payers.
Formally, if two scores $(u_1, v_1)$ and $(u_2, v_2)$ are such that 
$Score(u_1, v_1)=Score(u_2, v_2)$, $T(v_1)=T(v_2)$, 
and $R(u_1)\geqslant R(u_2)$, then $Conf(u_1, v_1)\geqslant Conf(u_2, v_2)$.

This prior knowledge incorporates the payer's intention in measuring the confidence of transaction scores. 
In this way, payees may have different confidence in receiving transactions with the same anonymity.
%This way the transactions with the same anonymity received by a payee may have different confidence. 
For instance, exchanges on Ethereum receive funds from payers with different motivations---some are ordinary investors and some are suspicious accounts.

Below, we propose the \textit{Confidence} formulation that satisfies the above items of prior knowledge:
\begin{equation}
    Conf(u, v) = \frac{R(u) + (1-|Score(u,v)-T(v)|)}{2}.
    \label{C}
\end{equation}
% \begin{equation}
%     Conf(u, v) = \frac{\alpha_1 R(u) + \alpha_2 (1-|Score(u,v)-T(v)|)}{\alpha_1 + \alpha_2},
% \end{equation}
% where the tunable parameters $\alpha_1$ and $\alpha_2$ are non-negative integers. 

Then, we describe how to quantify the \textit{Reliability} metric of a payer by the transactions it sends.

\noindent \textbf{[Prior knowledge 5] }
Payers with higher reliability send transactions with higher confidence. 
For two payers $u_1$ and $u_2$ with equal scores, %Please check intended meaning is retained.
%qs comment: then payer $A$ has a higher reliability 改成了payer $u_1$
if payer $u_1$ has higher confidence for all out transaction scores than $u_2$, then payer $u_1$ has a higher reliability.
Formally, if two payers $u_1$ and $u_2$ have $h:Out(u_1)\rightarrow Out(u_2)$ and $Conf(u_1,v_1) > Conf(u_2, h(v))$ $\forall(u_1,v)\in Out(u_1)$, then $R(u_1)>R(u_2)$. 
The corresponding formulation of \textit{Reliability} metric for $\forall u \in U$ is defined as 
\begin{equation}
    R(u) = \frac{\sum_{(u,v)\in Out(u)} Conf(u,v)}{|Out(u)|}.
    \label{R}
\end{equation}

%TODO
Finally, the risk rating of an account is calculated by $Risk(u) = (1-R(u))\times 10$. %See the pseudo code of \textit{RiskProp} in Section~\ref{sec:Supplement}.
The pseudo-code of \textit{RiskProp} network propagation is described in Algorithm~\ref{alg:riskprop}.
% \subsubsection*{\textbf{Initialization of metrics}} 
Let $T^{0}$, $Conf^{0}$, $R^{0}$ be initial values and $t$ be the number of interactions. In the beginning, we have initial reliability $R^{0}$ $\forall u \in U$, initial trustiness $T^{0}=0.5$ $\forall v \in V$, and initial confidence $Conf^{0}=0.5$ for all transactions. 
%qs comment:0.001改成了0.01
Then, we keep updating metrics using Equations 2--4 until $\Delta$ is less than 0.01.

% In the beginning, we initial trustiness $T^{0}=0.5$ for all accounts, and confidence $Conf^{0}=0.5$ for all transactions. The initial reliability $R^{0}$ depends on whether labels are given during the learning procedure. When labels are unavailable, we set $R^{0}(u)=0.7$, $\forall u \in U$, in which case we call unsupervised learning.

%Whereas, the case where labels are given is called supervised learning, and we initialize $R^{0}$ of training accounts with different values (see details in Section~\ref{subsec:Supervised}), and set $R^{0}=0.7$ for testing accounts.

% and set $R^{t}=R^{0}$ of labeled illicit accounts during the training procedure.
%For unsupervised learning, only the pre-processed payer-payee network is fed into our RiskProp and the labeled data are unavailable. In the supervised learning , we would initialize the \textit{reliability} metrics for the training accounts (see Section~\ref{subsec:Supervised}).

\subsubsection*{\textbf{\textit{RiskProp+}: A Semi-supervised Version}}
%很多时候，我们对一些欺诈账户（经过验证的，钓鱼诈骗等）和合法账户的标签有部分了解。我们可以利用这样的先验信息，并以半监督的方式将它们纳入我们的方法中。
\blue{Sometimes, we have partial information about the labels of fraudulent accounts (verified, phishing scams, etc.) and licit accounts. We can take advantage of such prior information and incorporate them into our approach in a semi-supervised manner.
In the semi-supervised \textit{RiskProp+}, we initialize the \textit{Reliability} metrics only for the training accounts. According to the risk levels of services reported by \textit{Chainalysis}~\cite{Chainalysis2021}, we set $R^0=0.9$ for \textit{ICO wallet}, \textit{Converter}, and \textit{Mining}, $R^0=0.7$ for \textit{Exchange}, $R^0=0.4$ for \textit{Gambling}, $R^0=0$ for \textit{Phish/Hack}, and set $R^{0}=0.7$ for testing accounts. The reliability values of labeled illicit accounts are unchanged during the training procedure. }

\subsubsection*{\textbf{Example}}
Here, we use a small real-world dataset on Ethereum to intuitively show the results of \textit{RiskProp+} after interactions.
We collect transactions of 10 accounts (6 for training and 4 for testing), including 28,598 accounts and 52,733 transactions in total.
Table~\ref{tab:example} shows how the reliability of the 4 testing accounts varies over interactions (we omit trustiness and confidence for brevity). These testing accounts have the same reliability values at the beginning ($R^{0}=0.7$). After convergence, accounts labeled as ``phish/hack'' get a lower value of reliability, and other licit accounts get  higher reliability. Confirming our intuition, \textit{RiskProp} learns that accounts 0xebdc and 0xfe34 are high-risk accounts that investors need to be aware of.

\subsubsection*{\textbf{Workflow for Account Risk Rating}}
\blue{Figure~\ref{fig:workflow} shows the workflow of account risk rating on Ethereum, which contains four modules: \textbf{(i) Data acquisition} collects accounts, transactions, and labels from Ethereum and \href{https://etherscan.io/}{Etherscan}. Only a few labels are provided, and these labels are not available in the unsupervised setting. 
\textbf{(ii) Data pre-processing} of raw transaction data described in Figure~\ref{fig:background} is conducted %Please check intended meaning is retained.
in two steps: de-anonymous score calculation and directed bipartite graph construction (i.e., payer--payee network).
\textbf{(iii) Account risk rating} recursively calculates the \textit{Reliability}, \textit{Trustiness} of accounts, and \textit{Confidence} of transaction scores until convergence, updated by the propagation mechanism. 
%qs comment: evolution改成了 evaluation
\textbf{(iv) Results analysis} contains risk rating results analysis, comparative evaluation, and further analysis.}

\begin{figure}
  \centering
  \includegraphics[width=0.9\linewidth]{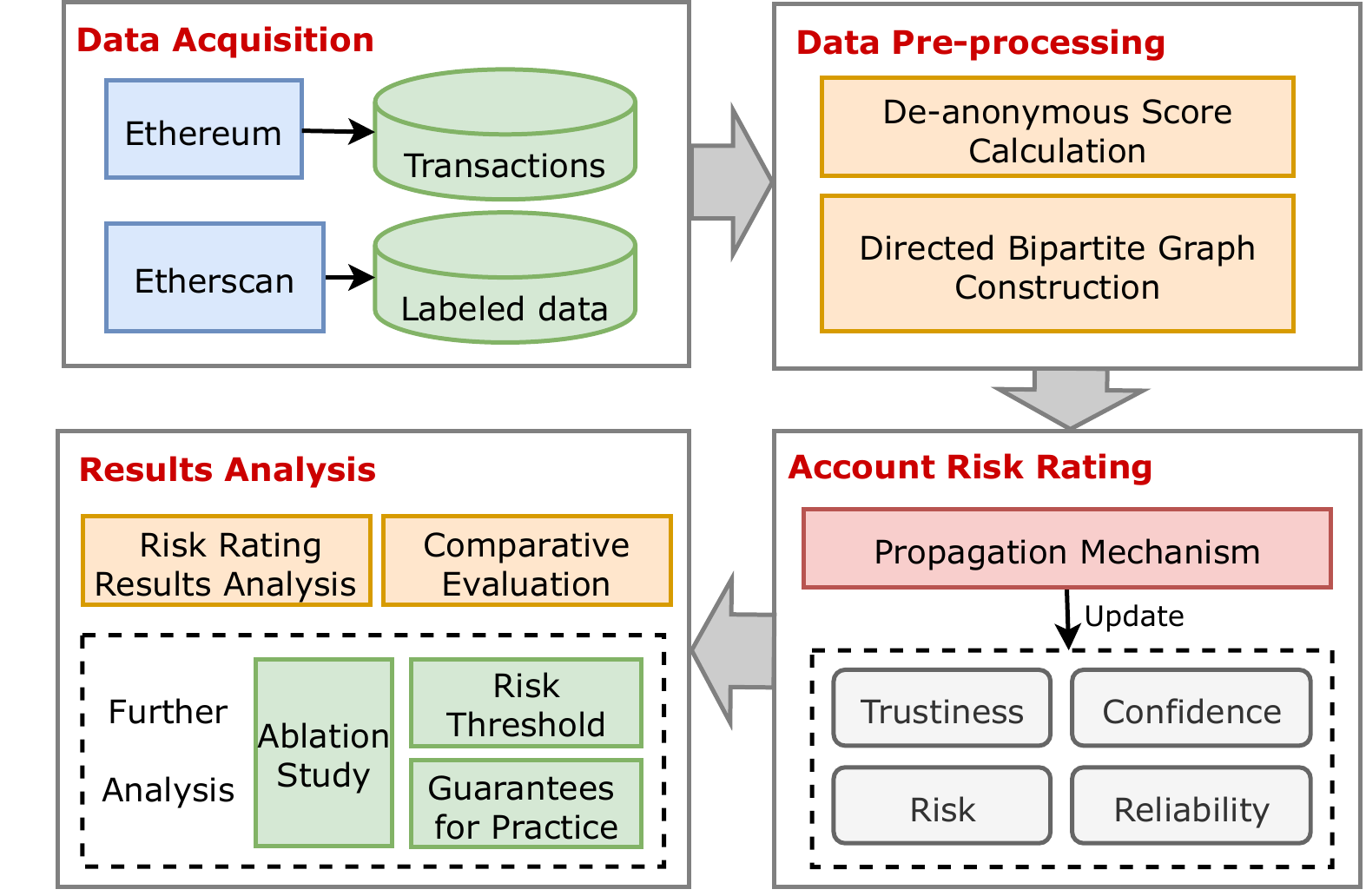}
  \vskip -2ex
  \caption{The workflow of account risk rating on Ethereum. }
  \label{fig:workflow}
  \vskip -2ex
  
\end{figure}

\vspace{-2ex}
\section{Experiments}
To investigate the effectiveness of \textit{RiskProp}, we conduct experiments on a real-world Ethereum transaction dataset. As risk rating is an issue without any ground truth, %we verify the effectiveness and significance of the risk rating results of \textit{RiskProp} via two tasks: 1) \textbf{risk rating analysis}, which includes statistical analysis of rating results and case studies of predicted high-risk accounts via manual verification; 2) \textbf{comparative evaluation}, which report on the classification performance of labeled accounts compared with baselines.
\blue{we verify the effectiveness and significance of the risk rating results of \textit{RiskProp} via three tasks: 1) \textbf{risk rating analysis}, which includes distribution of risk rating results and case studies of transaction pattern; 2) \textbf{comparative evaluation}, which reports on the classification performance of labeled accounts compared with various baselines; and  3) \textbf{further analysis}, which contains ablation study, impact of risk threshold, and guarantees for practical use.}
\textit{RiskProp} is open source and reproducible, and the code and dataset are publicly available after the paper is accepted.
%can be accessed at \url{https://www.dropbox.com/sh/gj0byzpmpivhern/AACG8IPgCL0M4bHwFb5Ie6BJa?dl=0}.

%\url{https://github.com/lindan113/RiskProp}.

%主任务是risk rating，因为它没有ground truth, 所以做了两方面的验证，一个是分类器（已有方法的验证）另一个是case study(人工验证)
%First, we present a theoretical analysis of the proposed de-anonymous Score. Next, we analyze the results of the risk rating by manually check historical records. we verify the performance on illicit ratings in both unsupervised and supervised settings. Finally, we report the case studies and investigate the transaction patterns of the predicted illicit accounts.

\vspace{-2ex}
\subsection{Data Collection}

%To train our model using supervised learning, w
We first obtain 803 ground truth account labels from an official Ethereum explorer and then include all the accounts and transactions that are within the one-hop and two-hop neighborhood of each labeled account. Next, we filter out the zero-ETH transactions and construct the records into a graph, retaining the largest weakly connected component for experiments. As a result, there are 1.19 million accounts and 4.13 million transactions in the network. 
In the dataset, \textbf{0.02 percent} (243) are labeled illicit (e.g., phishing scam), whereas \textbf{0.05 percent} (560) are labeled licit (e.g., exchanges). The remaining unknown accounts are not labeled with regards to licit versus illicit. 
%The dataset can be found at \url{https://1drv.ms/u/s!AjeNZgmClbKY6kAJtpRCOg0_M39D?e=XsVy8m}.
% 816
\vspace{-0.75ex}
\subsection{Effectiveness of De-anonymous Score}
% 方差分析就是对试验数据进行分析，检验方差相等的多个正态总体均值是否相等，进而判断各因素对试验指标的影响是否显著

%The one-way analysis of variance (ANOVA) is used to determine whether the mean values of different groups were significantly different in a quantitative way.
We use one-way analysis of variance (ANOVA) to assess whether there is a significant difference between illicit and licit transactions in the proposed de-anonymous score in Equation (\ref{equ:Score}). 
We consider a transaction as \textit{illicit} (versus \textit{licit}) if its payer is marked as \textit{illicit} (versus \textit{licit}).
\blue{Table~\ref{tab:anova_rand_das} shows that compared with the random score, our proposed score achieves a larger mean square (MS) between groups and smaller MS within groups; in addition, our proposed score has a higher F value, and the $p$-value equals $0$. These results suggest that the de-anonymous score is a useful metric for assessing the quality of transactions.}
%The one-way ANOVA results of the random Score of transactions and our proposed de-anonymous Score are described in Table~\ref{tab:anova_rand} and Table~\ref{tab:anova_ours}, respectively. Compared with the random Score, our proposed Score achieves a larger Mean Square (MS) between groups and smaller MS within groups, which suggests that there is a significant difference between the illicit and licit transaction in terms of de-anonymous Score. Furthermore, our proposed Score has a higher F value, indicating that different groups of samples are the main sources of variance, and the $p$-value of our proposed Score equals $0$, which also underlines that different categories of account have different impacts on the de-anonymous Score. Here we confirm that our proposed Score is a useful metric for assessing the quality of transactions.

\begin{table}[htbp]
  \caption{ANOVA of random scores and de-anonymous scores.}
  \centering
  \vskip -1.5ex
  \label{tab:anova_rand_das}
  \scalebox{0.7}{
  \begin{tabular}{c|ccc|ccc}
    \toprule
    &  \multicolumn{3}{c|}{Random scores} & \multicolumn{3}{|c}{ De-anonymous score} \\
    Src of var. & MS & F & $p$-value & MS & F & $p$-value \\
    \midrule
    \midrule
    Between groups
    &$8.8 \times 10^{-1}$ & $2.6 \times 10^1$ & $1.0 \times 10^{-1}$ 
    &$7.8 \times 10^2$  & $7.7 \times 10^3$ & 0 \\
    Within groups   
    &$3.3 \times 10^{-1}$&- &- 
    &$1.0\times 10^{-1}$ &- &-\\
  \bottomrule
  \end{tabular}
  }
    \vskip -1ex
\end{table}

% %修改了小数点不一致问题，统一成小数点后保留四位小数，并把一些较大的数改成了科学计数法表达
% \begin{table}
%   \caption{One-way ANOVA of random Scores.}
%   \centering
%   \vskip -1.5ex
%   \label{tab:anova_rand}
%   \scalebox{0.75}{
%   \begin{tabular}{c|ccccc}
%     \toprule
%     Src of var. & MS & F & $p$-value & \\
%     \midrule
%     \midrule
%     Between groups  %&0.87687 & 2.6337 & 0.104615 &  
%     %&0.8769 & 2.6337 & 0.1046\\
%     &$8.769 \times 10^{-1}$ & $2.634 \times 10^1$ & $1.046 \times 10^{-1}$\\
%     Within groups   %&0.332936 &- &-
%     %&0.3329 &- &-\\
%     &$3.329 \times 10^{-1}$&- &-\\
%   \bottomrule
%   \end{tabular}
%   }
% \vskip -1ex
% \end{table}

% % \vspace{-4pt}

% \begin{table}
%   \caption{One-way ANOVA of de-anonymous Scores.} 
%   \centering
%   \vskip -1.5ex
%   \label{tab:anova_ours}
%   \scalebox{0.75}{
%   \begin{tabular}{c|ccccc}
%     \toprule
%     Src of var. & MS & F & $p$-value &  \\
%     \midrule
%     \midrule
%     Between groups  
%     &$7.805 \times 10^2$  & $7.654 \times 10^3$ & 0 \\
%     Within groups  
%     &$1.020\times 10^{-1}$ &- &-\\
%   \bottomrule
%   \end{tabular}
%   }
% \end{table}

\subsection{Analysis of Risk Rating Results} %不谈有监督和无监督

The principal task of \textit{RiskProp} is to rate Ethereum accounts based on how ill-disposed they are. Given the account risk rating obtained by \textit{RiskProp}, we first review the results and investigate the capability of \textit{RiskProp} in discovering new risky accounts. Then, we dig deeper into the predicted high-risk accounts and obtain some insights.

\subsubsection{Distribution of risk rating results}

The risk value of an account ranges from 0 (low risk) to 10 (high risk). \blue{The distribution of the predicted risk scores is as follows: 33.58\% are located at (0,2], 63.45\% are located at (2,4], 2.03\% are located at (4,6], 0.78\% are located at (6, 8], and 0.19\% are located at (8, 10].}
%Table~\ref{tab:risk_count} shows the distribution of the predicted risk Scores via \textit{RiskProp}. The results shows that 99\% of accounts lie in (0, 6], and 0.19\% of accounts lie in (8, 10], 
This is consistent with expectations: The risk value of the Ethereum transaction network meets the power distribution law, indicating that the overwhelming majority of accounts act normally, and only very few accounts have abnormal behaviors.
%The risk rating of the account can be further divided into several levels based on their value. Specifically, (0, 2] denotes accounts with ``LOW'' risk, (2, 4] denotes accounts with ``MODERATELY LOW'' risk, (4, 6] denotes accounts with ``MODERATE'' risk, (6, 8] denotes accounts with ``MODERATELY HIGH'' risk, and (8, 10] denotes accounts with ``HIGH'' risk.
We are interested in whether the high-risk accounts predicted by \textit{RiskProp} are actually questionable.
Thus, we first manually check the top 150 accounts with the highest risk (with both in-coming and out-going transactions). %, as illustrated in Figure~\ref{fig:top100}
The finding is that 119 out of 150 (approximately \textbf{80\%}) accounts have abnormal behaviors. Among these 119 illicit accounts, 43 accounts are already labeled as ``phish/hack'' by Etherscan, whereas the remaining 76 are \textbf{newly discovered} suspicious accounts that are not marked in the existing label library. This result indicates the capabilities of \textit{RiskProp} in predicting undiscovered risky accounts and reducing financial losses.

\begin{figure}[t]
    \centering
    \includegraphics[width=0.9\linewidth]{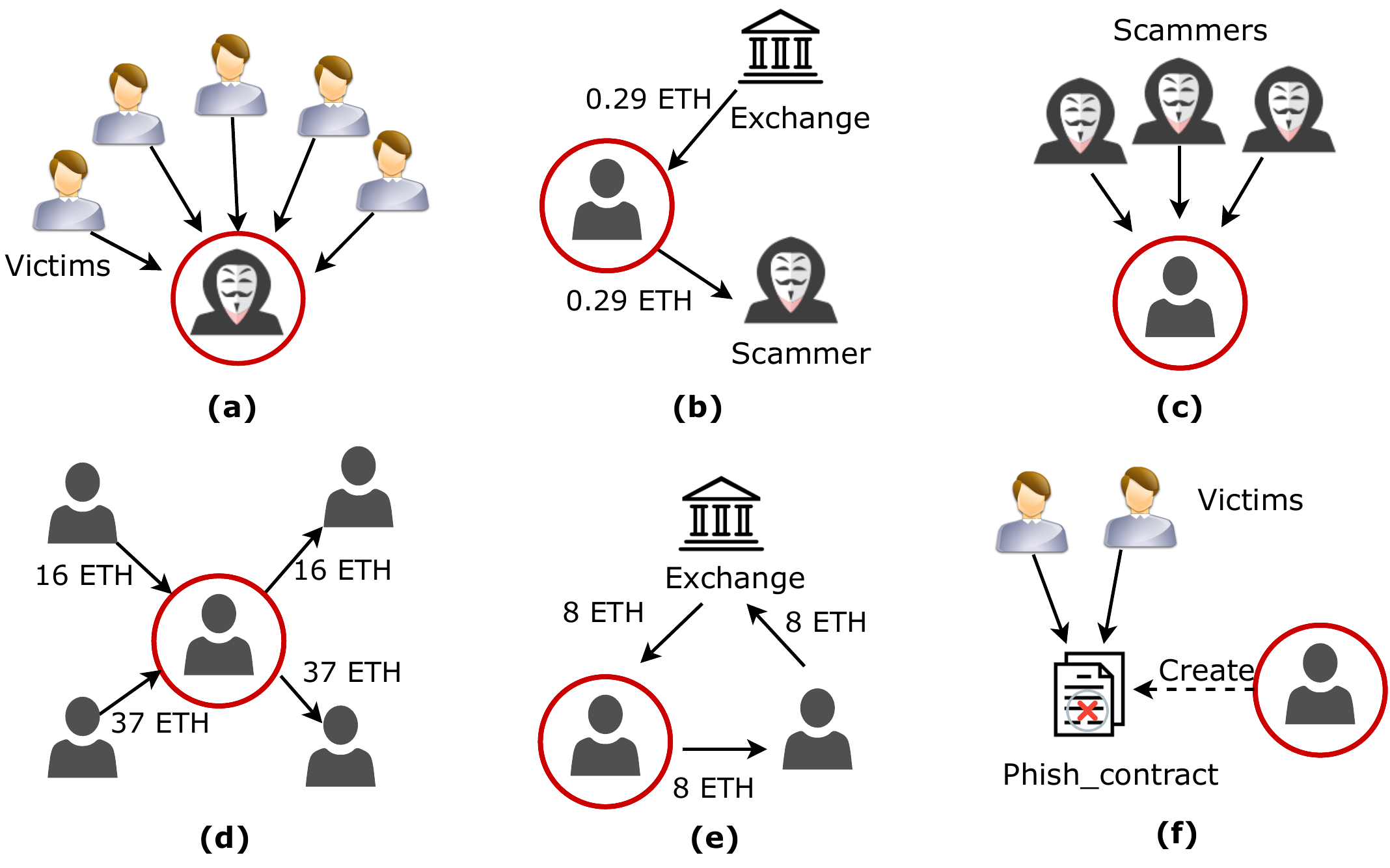}
    \vskip -1.5ex
    \caption{Visualization showing some typical transaction patterns of risky accounts (in red circles). }
    \label{fig:casestudy}
    \vskip -2ex
\end{figure}

% \subsubsection{Further analysis and case studies}
\subsubsection{Case studies of transaction pattern}
\label{subsec:case_studies}
We then manually verified the predicted risky accounts by investigating their abnormal behaviors and find that there are many suspicious transaction patterns in the network.
In order to save space, we show 6 typical patterns in Figure~\ref{fig:casestudy}. These patterns are summarized from the real-world Ethereum transaction data and guided by current research and recommendation reports.

%qs comment: add the:  is the Bee Token ICO Scam
    \textbf{(a) Hacking scammers} are a list of addresses related to phishing and hacks. Figure~\ref{fig:casestudy}(a) shows a pattern of phishing accounts reported by users who suffered financial loss. A typical phishing scam on Ethereum is the \href{https://theripplecryptocurrency.com/bee-token-scam/}{``Bee Token ICO Scam” attack}, in which the phishers sent fake emails to the investors of an ICO with a fake Ethereum address to deposit their contributions into. For example, account \href{https://cn.etherscan.com/address/0xe336327426b8f95a5f5eb1f74144fd9065069c28} {0xe336} has been confirmed to be part of this ``Bee Token'' scam, and 243 ETH  has been sent to this address by 165 victims.
    
    \textbf{(b) Fund source of hacking scammers} are the upstream accounts of the known illicit accounts, which are collusion scam accounts to attract victims or provide money for hacking. As shown in Figure~\ref{fig:casestudy}(b),
    the behaviors of collusion scam accounts may look similar to victims. Nevertheless, we find that the upstream collusion accounts appear to participate in fewer transactions  with shorter time intervals, and there are attempts to transfer the entire ETH balance of the scammers according to the \textit{Red Flag Indicators} of FATF~\cite{fatf2022}. 
    % A new user attempts to trade the entire balance of VAs, or withdraws the VAs and attempts to send the entire balance off the platform
    % Red Flag Indicators Related To Transaction Patterns
    % and it is hard to distinguish them based solely on the connection with scammers

    \textbf{(c) Money laundering of scammers} are the downstream accounts of the known illicit accounts, which are collusion scam accounts to accept and transfer the stolen money, obfuscating the true sources. 
    As shown in Figure~\ref{fig:casestudy}(c), account \href{https://cn.etherscan.com/address/0x78f1e9312270d9acfc63c42910b352a6b759cd9d}{0x78f1} received stolen funds from several known hacking scammers, appearing to be the account used in the ``placement'' stage of money laundering.
    Another example is \href{https://cn.etherscan.com/address/0xcfddced74ac351d39f64ee8bc1eab83ae2c9fd4b}{0xcfdd}, which receives stolen funds from the Fake Starbase Crowdsale Contribution account \href{https://cn.etherscan.com/address/0x122c7f492c51c247e293b0f996fa63de61474959}{0x122c}.
    % 0x78F1e9312270D9ACFC63C42910b352A6B759Cd9d 多个phishing的下游
    %qs comment: are that the middle accounts serve改成了are the middle accounts that serve
    \textbf{(d) Zero-out middle accounts} are the middle accounts that serve as a bridge defined by Li \textit{et al.}~\cite{flowscope2020}. As shown in Figure~\ref{fig:casestudy}(d), most of the received funds will be transferred out in short succession (such as within 24 hours). See \href{https://cn.etherscan.com/address/0x126e75ca0340d84f59dff8aa199d965109a2e639}{0x126e} for an example.

    \textbf{(e) Round transfers among exchanges} denote a pattern that an account withdraws ETH without additional activity to a private wallet and then deposits back to the exchange, as shown in Figure~\ref{fig:casestudy}(e). Account \href{https://cn.etherscan.com/address/0x886e8902916752b360ded1993263e9791d20e9f6}{0x886e} withdraws 0.4 ETH from Cryptopia exchange and then deposits the same amount of ETH back to Cryptopia, which is an unnecessary step and incurs transaction fees~\cite{fatf2022}. Such a phenomenon indicates that the exchange is misused as a money-laundering mixer or is conducting wash trading~\cite{victor2021detecting}.
    
    % Depositing VAs at an exchange and then often immediately –
    % o withdrawing the VAs without additional exchange activity to other VAs, which is an unnecessary step and incurs transaction fees;
    % o  withdrawing the VAs from a VASP immediately to a private wallet. This effectively turns the exchange/VASP into an ML mixer.
    
    \textbf{(f) Creators of illicit contracts} are often the manipulators behind the scenes.
    The Origin Protocol phishing scam contact account \href{https://cn.etherscan.com/address/0x98198d7ecc42a3e18d3fc52957db46930cf1d2e8}{0x9819} was created by account \href{https://cn.etherscan.com/address/0xff1a418a203c812ffa85cd3dd52776727258f22f}{0xff1a}. After victims deposited money into the phishing contract, the creator transfers the stolen funds back to himself via internal transactions, which deliberately enhances anonymity.

We observe that many illicit accounts are outside the label library and are still considered risk-free. Based on the results, we infer that our \textit{RiskProp} is able to expose  unlabeled illicit accounts. This is crucial on Ethereum, which lacks authorized and effective regulation. In addition, the newly identified illicit accounts can complete the current label collection for additional analysis.
% 可以实现标签补全，可以还有很多异常账户是没有被标记的
\vspace{-1ex}

\subsection{Comparative Evaluation Settings}
\label{subsec:Classification}

%As we mentioned, the main aim of this work is account risk rating on Ethereum. However, to the best of our knowledge, there exists no other Ethereum risk rating method in existing literature, and no ground truth of the risk rating Scores to directly verify the effectiveness our results. Therefore, 
To further evaluate the performance of our method and show the potential application, we employ the rating scores to conduct classification experiments that
divide Ethereum accounts into illicit and licit accounts, and we compare the results with the existing baseline methods for further verification. We wish to investigate if \textit{RiskProp} can give a higher risk rating for the known illicit accounts and a lower rating for known licit accounts.
% 做的是评分系统，市面上没有其他的区块链的风险评分，并且没有评分结果的ground truth,所以 对评分做了个分类，并和现有的分类器的结果做对比，进一步去验证我们的方法的有效性

\subsubsection{Compared Methods} 
\label{subsec:compared} 
% 审稿人3:评价部分比较方法的选择不够有说服力。我不清楚这些比较的方法是否适合这个应用程序。如果这些方法一开始不合适，评价结果就没有多大意义。

\blue{As mentioned earlier, \textit{RiskProp} is the first algorithm that explores the risk rating of blockchain accounts. We chose a variety of methods (unsupervised and supervised) as baselines, which are similar to the problem we want to solve. We compare unsupervised \textit{RiskProp} with (i) \textbf{web page ranking}, such as PageRank~\cite{BRIN1998107}, and (ii) \textbf{bipartite graph-based fraud detection,} such as FraudEagle~\cite{akoglu2013opinion}, BIRDNEST~ \cite{hooi2016birdnest}, and REV2~\cite{Kumar2018rev2}, which are also unsupervised methods.}

\blue{The (semi-)supervised approaches are as follows. 
\textbf{(i) Machine learning methods}, e.g., logistic regression (LR), na\"ive Bayes (NB), decision tree (DT), support vector machine (SVM), random forest (RF), extreme gradient boosting (XGBoost), and LightGBM. These methods are used by  \cite{Al-Emari2021,Agarwal2021,Chen2020Phishing,Li2022TTAGN} for detection of abnormal Ethereum accounts. 
\textbf{(ii) Traditional graph neural network}, including DeepWalk, Node2Vec, and graph convolutional network (GCN) were conducted by Chen \textit{et al.}~\cite{Chen2021} for detection of Ethereum phishing scams.
\textbf{(iii) Graph neural network for graphs with heterophily}, such as CPGNN~\cite{zhu2021cpgnn}. The application of this type of algorithms is a recent research advancement in the task of Ethereum account classification~\cite{huang2022ethereum}.}  %. on Ethereum transaction networks, scammers are more likely to build connections with normal accounts instead of other scammers, thus heterophilic GNN may be suitable for our dataset

\subsubsection{Evaluation Metrics}

To evaluate the performance of the models, we calculate the following metrics: $Precision$, $Recall$, $F1$, $Accuracy$, and $AUC$.
As we know, there are only 6 out of 10,000 (0.067 percent) accounts labeled in the entire dataset. To measure the order of the risk rating, we employ $Precision@k$ and $Recall@k$ to evaluate the ranking order of the algorithm ($@k$ means the top $k$ accounts). %See Section~\ref{subsec:settings} for details).
All baseline methods are tested using the original codes published by the authors. We repeat experiments 10 times and report the average results.

\begin{figure}
	\subfigure{%Illicit account rating
			\centering
			\includegraphics[width=0.46\linewidth]{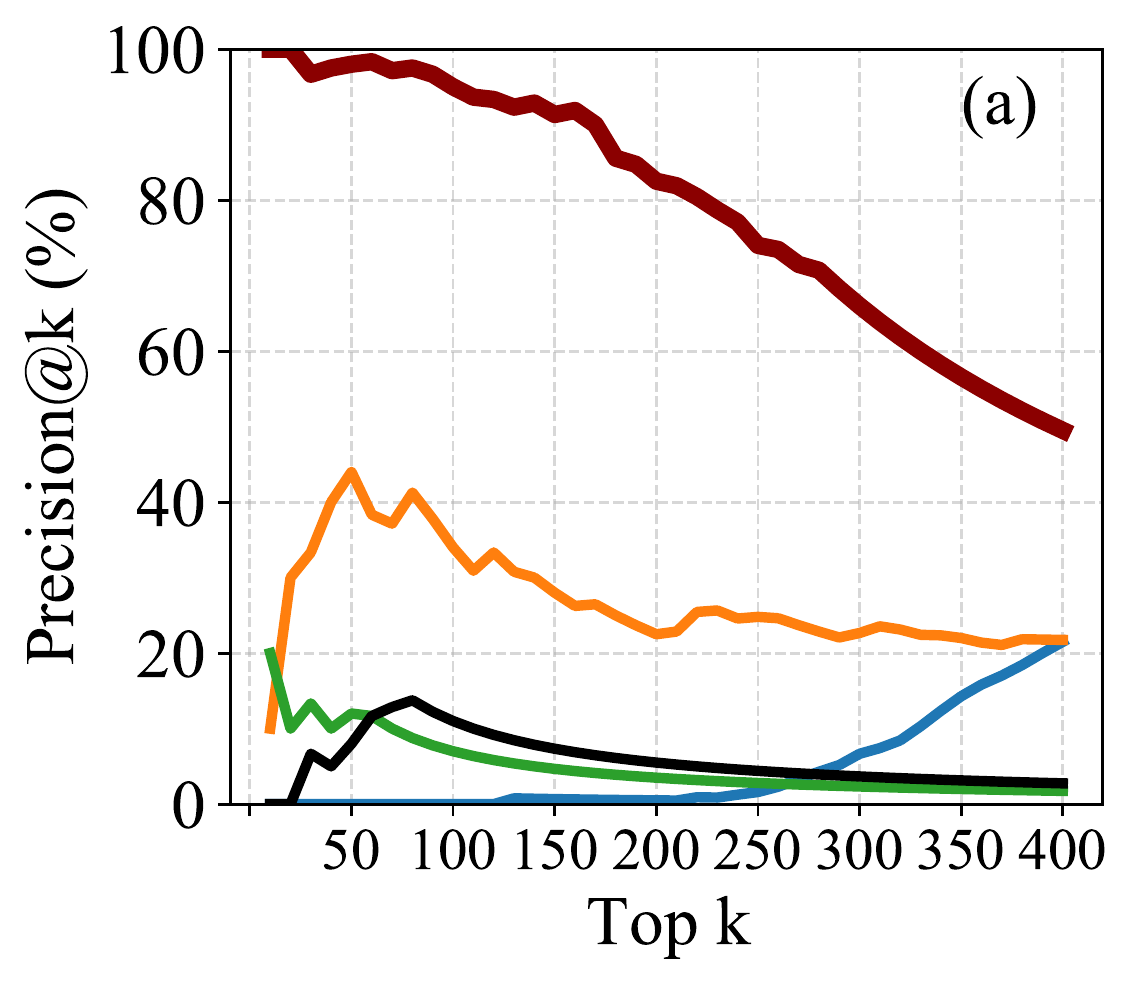}
	}    
	\subfigure{%Illicit account rating
			\centering
			\includegraphics[width=0.46\linewidth]{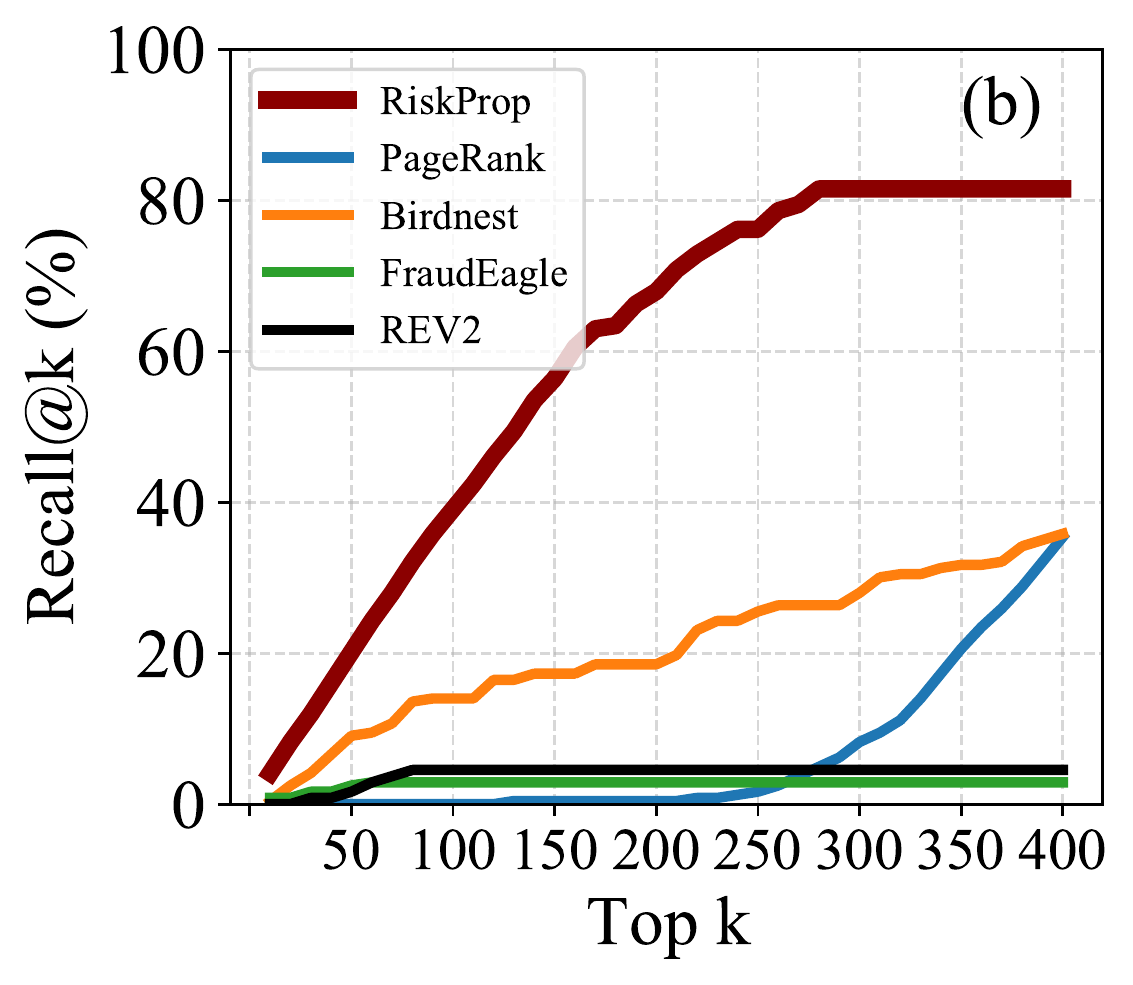}
	}    
	\vskip -1.5ex
	\caption{\blue{The $Precision@k$ and $Recall@k$ of illicit account prediction with different rating methods.}}
	\vskip -2ex
	\label{fig:unsupervised}
\end{figure}

\subsubsection{Implementation Details}
We evaluate the methods with binary labeled accounts (\textit{illicit} verse \textit{licit}) and, thus, we assume accounts in the top 1\% to be the illicit accounts (corresponding threshold: 6 for \textit{RiskProp}). The reason %of
{for} this threshold and percentage setting is discussed in Section~\ref{subsec:parameter analysis}. 
The split of the dataset in the (semi-)supervised setting is $training:test = 8:2$.
%  and 2.3e-06 for PageRank
% In comparison to those classifiers, we consider an account suspicious if its risk rating is greater than a predefined risk threshold $RTH$. In this case, we set $RTH=6$. Section~\ref{subsec:parameter analysis} explains why we chose 6 as $RTH$. 

\vspace{-2ex}
\subsection{Comparative Evaluation Results}

\blue{We report the $Precision@k$ and $Recall@k$ curves of the compared algorithms, as shown in Figure~\ref{fig:unsupervised}. We observe that \textit{RiskProp} obtains superior precision and recall than that of baseline with different $k$. Up to $k=100$, the precision of RiskProp is almost 1 for illicit account prediction, which is surprising for an unsupervised setting. The $Recall@k$ curve of RiskProp is significantly higher than the compared methods, and also increases steadily with $k$. 
Table~\ref{tab:results} shows the performance of unsupervised and supervised methods separately. 
We observe that \textit{RiskProp} remarkably outperforms the unsupervised graph rating baselines in terms of accuracy and AUC, improving by 38.90\% and 34.16\%, respectively.
Meanwhile, for the licit account prediction, we observe that RiskProp beats the best baseline (FraudEagle) with a 10.48\% improvement in its F1-score.  
These demonstrate the effectiveness of our account risk rating method without labeling information.}

% 

% Then, we compare unsupervised \textit{RiskProp} with PageRank (shorts for ``PR''), which is a famous network ranking algorithm. 
% First, our method obtains superior precision and recall than that of baseline with different $k$. The proposed \textit{RiskProp} obtains 80\% recall even when $k=800$ for illicit account prediction and achieves 91\% precision for licit account prediction, which is surprising for an unsupervised setting. 
% Second, the proposed \textit{RiskProp} significantly outperforms the baseline in terms of both $Precision@k$ and $Recall@k$. At $k=800$, our recall is 14\% higher than PageRank in detecting illicit accounts, and our precision is 25\% higher than PageRank in predicting licit accounts. These prove the effectiveness of our account risk rating methods without label information.

\blue{Next, we turn our attention to the results of the semi-supervised \textit{RiskProp+} compared with the existing (semi-)supervised classification in Table~\ref{tab:results}, from which we derive the following conclusions: }
1) \textit{RiskProp+} outperforms all baseline methods by 12.32\% in terms of F1-score, 13.62\% in terms of AUC, and 10.56\% in terms of accuracy. 
2) The precision of licit accounts prediction is improved from 82.52\% (i.e., the average precision in baselines) to 93.33\%, which means more licit accounts can be correctly identified.
% qs comment: in prediction中间加了the
3) The superior performance of \textit{RiskProp} is more significant in the prediction of illicit accounts. The recall of illicit accounts prediction is improved from 60.56\% (i.e., the average recall of illicit accounts prediction in baselines) to 84.78\%. This shows the effectiveness of our framework in the prediction of both illicit and licit accounts.

\begin{table}
  \centering
  \caption{The classification results (\%) of unsupervised and (semi-)supervised methods. }
  %The best result of each column is boldfaced.
  \vskip -1ex
  \label{tab:results}
  \scalebox{0.78}{
  \begin{tabular}{c|ccc|ccc|cc}
    \toprule
        & \multicolumn{3}{c|}{\textbf{Illicit account}} & \multicolumn{3}{|c}{\textbf{Licit account}} & \multicolumn{2}{|c}{\textbf{Total}} \\
        \textbf{Methods} & P & R & F1       & P & R & F1 & Acc.  & AUC  \\
    \midrule
    \midrule
        PageRank &29.13  &67.49  &40.69  &67.08  &28.75  &40.25  &40.47  &48.12    \\
    \midrule    
        FraudEagle & 12.28 & 2.88 & 4.670  & 68.36 & \textbf{91.07} & 78.10  & 64.38 & 46.98 \\
        % FraudEagle & 24.56 & 5.76 & 9.33  & 69.30 & \textbf{92.32} & 79.17  & 66.12 & 49.04 \\
        BIRDNEST &22.24  &47.32  &30.26  &55.24  &28.21  &37.35  &34.00 &37.77       \\
        % BIRDNEST &32.54  &50.62  &39.61  &71.76  &54.46  &61.93  &53.30 &53.30       \\
        REV2  &14.10 &4.527 &6.854 & 68.00 &88.04 &76.73  &62.76 &46.28   \\
        % REV2  &69.28 &79.83 &74.18 & 90.63 &90.63 &87.53  &83.18 &82.23   \\
    \midrule    
        RiskProp &\textbf{71.48}  &\textbf{71.48}  &\textbf{76.15}  &\textbf{91.44}  &85.89  &\textbf{88.58}  &\textbf{84.56}  &\textbf{83.69}    \\
    \bottomrule
    \toprule
        & \multicolumn{3}{c|}{\textbf{Illicit account}} & \multicolumn{3}{|c}{\textbf{Licit account}} & \multicolumn{2}{|c}{\textbf{Total}} \\
     \textbf{Methods} & P & R & F1       & P & R & F1 & Acc.  & AUC  \\
    \midrule
    \midrule
         LR  & 65.67 & 74.58 & 69.84         & 83.87 & 77.23 & 80.41 & 76.25  & 75.90  \\
         NB   & 59.79 & \textbf{98.31} & 74.36           & \textbf{98.41} & 61.39 & 75.61 & 75.00 & 79.85   \\
         DT & 62.66 & 54.07 & 58.04                & 75.79 & 81.19 & 78.40 & 71.75 & 68.39   \\
         SVM & \textbf{90.00} & 45.76 & 60.67        & 75.38 & \textbf{97.03} & 84.85 & 78.12 & 71.40  \\
         RF & 71.52 & 53.39 & 61.14                & 75.55 & 86.93 & 80.84 & 74.00 & 69.40  \\
         XGBoost & 67.35 & 55.95 & 61.11                      & 76.58 & 84.16 & 80.19 & 70.05 & 73.75  \\
         LightGBM & 75.77 & 65.19 & 69.93               & 84.23 & 92.57 & 88.21 & 81.86 & 77.75 \\
    \midrule     
         DeepWalk  & 66.85 & 66.30 & 66.54               & 86.48 & 86.75 & 86.61 & 83.13 & 81.03 \\
         Node2Vec  & 62.36 & 63.26 & 62.76               & 85.10 & 84.56 & 84.82 & 78.13 & 72.78  \\
         GCN & 20.83 & 27.78 & 23.81       & 79.46 & 68.99 & 73.86 & 60.63 & 47.40  \\         
    \midrule
        CPGNN & 52.17 & 61.54 & 56.47   &86.84  &81.82  &84.26  &76.88 & 71.68 \\
    \midrule    
        RiskProp+ & 70.91 & 84.78 & \textbf{77.23}    & 93.33 & 85.96 & \textbf{89.49} & \textbf{85.63} & \textbf{85.37}  \\
    \bottomrule
\end{tabular}
}
\vskip -1.5ex
\end{table}

\vspace{-1.5ex}
\subsection{Ablation Study}
%\subsubsection{Ablation Study}
To further validate the contribution of each component of the proposed \textit{RiskProp+}, we conduct an ablation study as follows.
% understand the power of the proposed \textit{RiskProp} method and
%Our approach consists of de-anonymous Score calculation and network propagation procedure. To understand how these modules impact performance, 
    
\noindent$\bullet$ \textbf{RiskProp+ (Full model):} All components of the model and label data are included.
    
\noindent$\bullet$ \textbf{w/o label:} Labels are unavailable in the learning procedure, and the model is trained in an unsupervised manner.
    
\noindent$\bullet$ \textbf{w/o network propagation (NP):} Remove the NP procedure and calculate the average de-anonymous scores ($ADS$) for each accounts' outgoing transactions (payer role). An account is predicted as abnormal if its $ADS\leqslant0$.
    
\noindent$\bullet$ \textbf{w/o de-anonymous score (DS):} Replace DS with random scores, ranging from $-$1 to 1.
%这里的标点符号之前是;齐双改成了。

\vskip -1ex
\begin{table}
    \centering
    \caption{Illicit account prediction of ablation studies. } %The abbreviation ``w/o'' stands for ``without''.
    \scalebox{0.8}{
    \begin{tabular}{c|ccc}
    \toprule
        Methods & Precision & Recall & F1-score\\
        \midrule
        \midrule
        RiskProp+ & 0.7091 & 0.8478 & \textbf{0.7723} \\
        RiskProp+ (w/o label) & \textbf{0.7148} & 0.8148 & 0.7615 \\
        RiskProp+ (w/o NP) & 0.3811 & \textbf{0.9959} & 0.5513 \\
        RiskProp+ (w/o DS) & 0.4737 & 0.1957  & 0.2769 \\
        \bottomrule
    \end{tabular}
    }
    \label{tab:ablation}
    \vskip -1.5ex
% \vspace{-2.1ex}
\end{table}

%qs comment: 加冠词has a only lower precision but a greatly improved recall
%Table~\ref{tab:ablation} shows the effectiveness of each procedure. 
We derive the following findings from Table~\ref{tab:ablation}:
1) Without the labels, the F1-score drops only slightly, indicating that our \textit{RiskProp} does not rely on label data and can obtain good results in an unsupervised manner. To our surprise, the full model outperforms the RiskProp (w/o label), with a 3.3\% increase in recall and 0.51\% decrease in precision. This may be possibly explained by the reliability values of labeled illicit accounts remaining unchanged during training in the supervised setting.
2) RiskProp (w/o NP) has a only lower precision but a greatly improved recall, revealing that most of the illicit accounts are correctly predicted as illicit but that some licit accounts are misjudged to be illicit. This result demonstrates that de-anonymous score is an effective indicator of illicit transactions but their confidence varies among transactions. This result also confirms why we need to consider the confidence of the score in the propagation mechanism.
3) RiskProp (w/o DS) yields low precision (47.37\%) and severely low recall (19.57\%). This result demonstrates that, even if the network propagation model is retained, the wrong scores of transactions will be spread throughout %Please check intended meaning is retained.
the entire Ethereum transaction network, resulting in poor prediction results.

%This shows that using de-anonymous Score or network propagation separately is not adequate for the task of account rating. 
%Specifically, unsupervised RiskProp (w/o DS) results in extremely low precision (17.28\%) and low recall (56\%). This indicates that even if the network propagation model is retained, the wrong Scores of transactions will be spread in the entire Ethereum transaction network, leading to poor prediction results.
%Moreover, unsupervised RiskProp (w/o NP) only hurts the precision but significantly benefits the recall, showing that most of the illicit accounts are correctly predicted as illicit, but the licit accounts are misjudged to be illicit accounts. This proves the de-anonymous Score is an effective signal for illicit transactions, but the confidence of de-anonymous Scores varies among transactions. This result also confirms why we need to consider the confidence of Score in propagation mechanism.
%Furthermore, we compare the unsupervised RiskProp with supervised RiskProp to investigate the effectiveness of labeled data. To our surprised, the supervised RiskProp outperforms the unsupervised one with 3.3\% increased in recall and 0.51\% decreased in precision. The possible reason may be the reliability values of labeled illicit accounts remain unchanged during training in the supervised setting, which leads to higher recall and lower precision. In addition, labeled accounts are only the tip of the iceberg for the entire transaction network, resulting in limited effect. \red{(Need modification!!!!!} %说明 我们的方法不依赖Label，结果也不错

\subsection{Risk Threshold of \textit{RiskProp}} 
\label{subsec:parameter analysis}
%qs comment: 第一句话改成Given accounts in the reversed order of their risk ratings
Given accounts in the reversed order of their risk ratings, a natural question is \textit{how to classify licit or illicit accounts according to their risk ratings for the classification task?}
One possible option is to determine the percentage of known illicit labels in the dataset and set the risk value of this percentage as a demarcation line for account classification. However, the percentage is imprecise because some of the illicit accounts remain unrevealed according to the experimental results in Section~\ref{subsec:case_studies}.
Therefore, we try to establish a suitable risk threshold ($RTH$) of RiskProp by conducting classification experiments.
Figure~\ref{fig:scalability_RTH}(a) demonstrates the results of illicit account prediction with different risk thresholds, %ranged
ranging from 1 to 10. As expected, the precision increases while the recall decreases with increasing $RTH$. In addition, F1, accuracy, and AUC first increase and then decrease with the increase in $RTH$. 
The best performance for F1 and AUC is when $RTH=6$. Thus, we set the risk threshold $RTH=6$ for \textit{RiskProp}.

%\vspace{-2ex}
\subsection{Guarantees for Practical Use} 
% TODO
Here, we present guarantees for \textit{RiskProp} in practical use regarding the following aspects: 1) guarantee of convergence; 2) time complexity; and 3) linear scalability. 

% Proof for prior knowledge
\noindent\textit{\blue{\textbf{Convergence and Uniqueness.}}}
\blue{We present the theoretical properties of \textit{RiskProp}, including the proofs of prior knowledge, convergence, and uniqueness of the proposed metrics, i.e., \textit{reliability}, \textit{trustiness}, and \textit{confidence}. Proofs are shown in the Appendix due to lack of space.}
% 直接指引去看附录好了，因为现在已经补充了证明
% Suppose $R^{\infty}(u)$, $T^{\infty}(v)$, and $Conf^{\infty}(u,v)$ are the \blue{final and unique} value of a payer $u$, a payee $v$, and a transaction $(u,v)$, respectively, %Please check intended meaning is retained.
% after convergence.
% % \subsubsection*{Convergence.}
% % The difference between $u$'s final reliability value and its value after the first interaction is at most 1, i.e., $|R^{\infty}(u) - R^{1}(u)| \leqslant 1$. Similarly, $|T^{\infty}(v) - T^{1}(v)|\leqslant 1$ and $|Conf^{\infty}(u,v) - Conf^{1}(u,v)|\leqslant 1$.
% %[Convergence theorem] 
% The maximum deviation of reliability of a payer at any iteration $t$ from its final value is bounded by\blue{ $|R^{\infty}(u) - R^{t}(u)|\leqslant 1, \forall u \in U$. As $t$ increases, $R^{t}(u)$ converges to $R^{\infty}(u)$. 
% Similarly, $|T^{\infty}(v) - T^{t}(v)|\leqslant ~$|S|$, \forall v \in V$, $|Conf^{\infty}(u,v) - Conf^{t}(u,v)|\leqslant \frac{1+|S|}{2}$,  $\forall (u,v) \in S$. See appendix~\ref{sec:proof} for detailed proof.}

\noindent\textit{\textbf{Time complexity.}}
In each interaction, the \textit{RiskProp} updates the \textit{reliability}, \textit{Trustiness} metrics of accounts and \textit{Confidence} metric of transactions.
Therefore, the complexity of each iteration is $\mathcal{O}(|U|+|S|)=\mathcal{O}(|S|)$, $|S|$ is the total edges in the payer--payee network. Thus, for $k$ iterations, the total running time is $\mathcal{O}(k|S|)$.

\noindent\textit{\textbf{Linear scalability.}} 
We have shown that \textit{RiskProp} is linear in running time in the number of nodes. To show this experimentally as well, we create random networks of an increasing number of nodes and edges and compute the running time of the algorithm until convergence. 
Figure~\ref{fig:scalability_RTH}(b) shows that the running time increases linearly with the number of nodes in the network. Therefore, we can conclude that \textit{RiskProp} is a scalable rating method %which
that is suitable for applications on large-scale transaction networks.

\begin{figure}
    \centering
    \subfigure[Impact of different \textit{RTH}]{
			%\centering
			\includegraphics[width=0.4\linewidth]{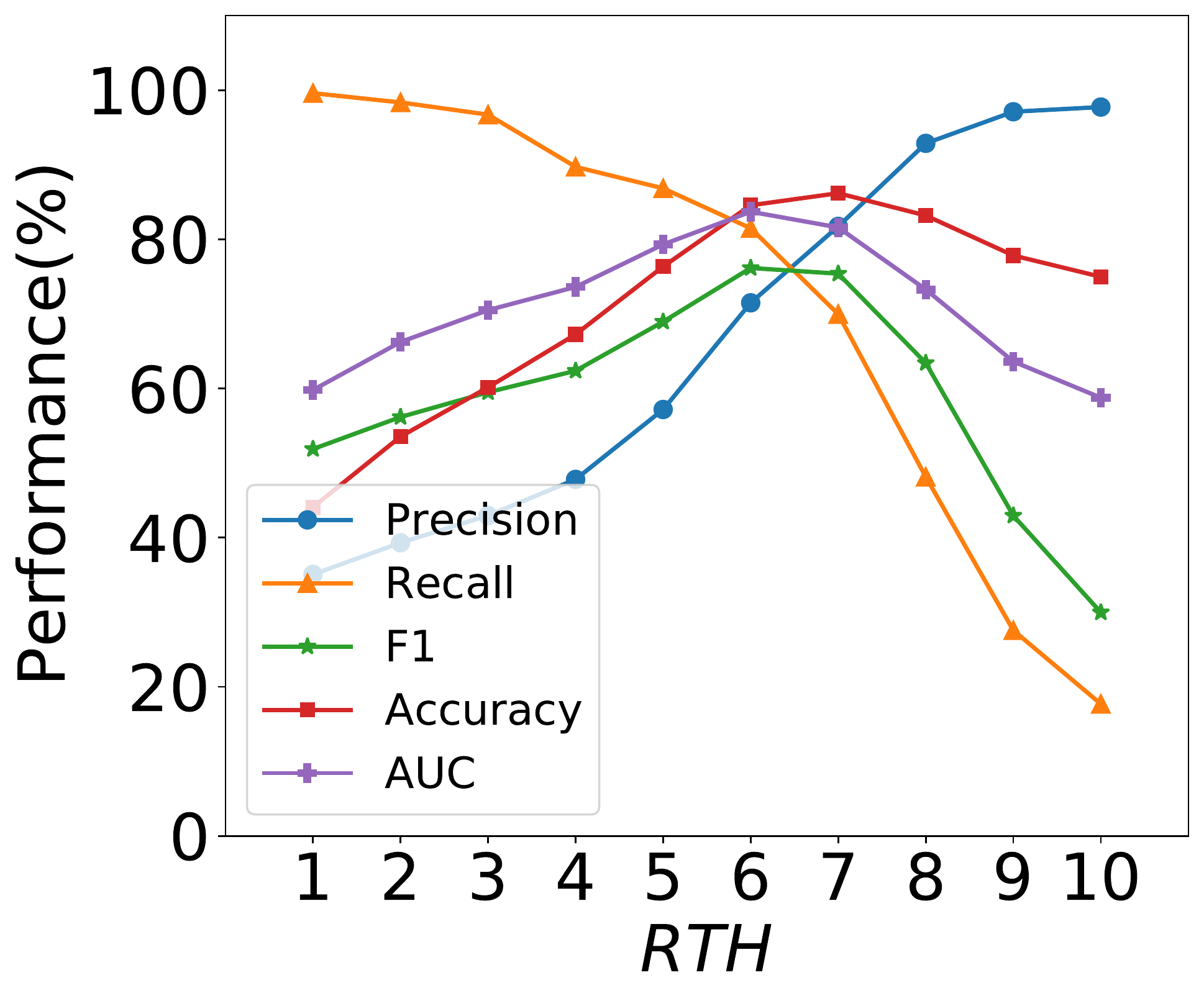}
	}
	\subfigure[Scalability of \textit{RiskProp}]{
        \includegraphics[width=0.4\linewidth]{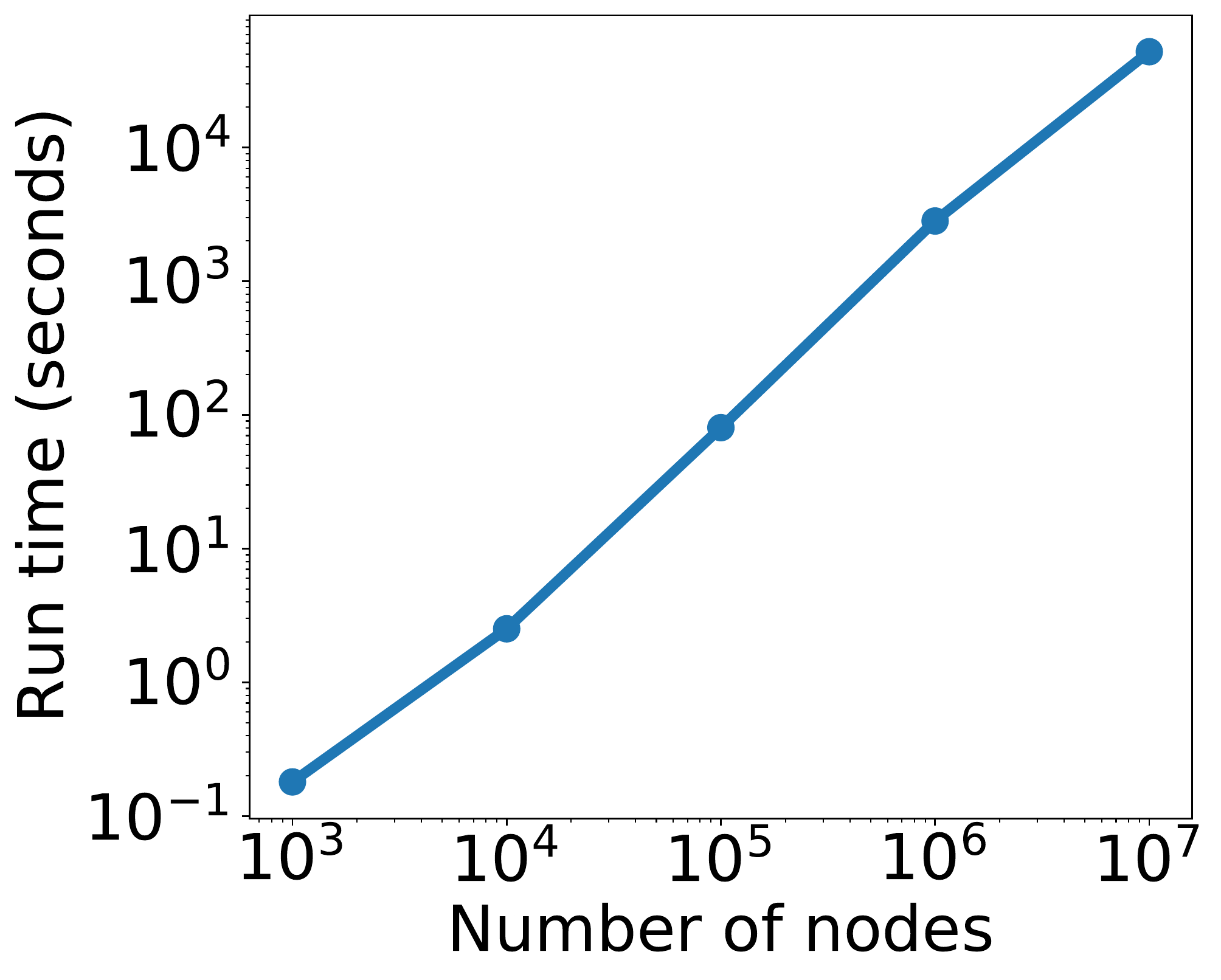}
    }
    \vskip -1.5ex
    \caption{Further analysis of \textit{RiskProp}.}
    \label{fig:scalability_RTH}
    \vskip -1.5ex
\end{figure}

\noindent\textit{\textbf{Analysis of incorrect predictions.}} 
\blue{Furthermore, the results of the \textit{RiskProp+} experiment showed that 39 out of 46 (85\%) phishing accounts were correctly predicted as illicit accounts. To understand why the remaining accounts failed to be detected as illicit by our model, we manually checked their transactions and neighbors and obtained the following results:
% 更进一步，我们从RiskProp+的实验结果看到，46个钓鱼账户中的39个钓鱼账户被我们正确预测为非法账户（i.e. 风险阈值超过6）.为了了解为什么剩下的7个账户未能被我们的模型检测为非法账户，我们手动检查了他们的交易和邻居。研究结果如下：
(i) For one of the accounts, we have a risk score of 5.72, and in practice, the system will also warn about such accounts that are close to the risk threshold ($RTH=6$).
% 1. 一个账户虽然没有被RiskProp+划分为非法账户，但是我们对它的风险评分为5.72，在实际应用中，这类接近风险阈值的账户，系统也将进行警告。
% 761：4
(ii) One account is set as the default risk value because the phishing account has no outgoing transactions for the time being, and in practice, we can make correct predictions as soon as the phishing account starts laundering money.
% 2. 有一个账户是因为该钓鱼账户暂时没有outgoing交易而被设为默认风险值，在实际应用中，只要该钓鱼账户开始洗钱，我们就可以进行正确的预测。
% 743：1
(iii) Among the remaining five accounts, one account has a high transaction volume of 154. The remaining four accounts have a high volume of transactions along with withdrawals of ETH from exchanges, which directly contributed to the high de-anonymity score of transactions. However, there are two sides to the story: regulation and fraud are a game of confrontation. For hackers, reusing accounts reduces the probability of being identified as high risk and, at the same time, reusing accounts and withdrawing money from exchanges increase the risk of exposure and fund freezing.}
% 3. 剩余的五个账户中，有个账户的交易量高达154次。剩下4个账户的交易量偏多，同时还存在从交易所提币的行为，这直接导致我们预先计算的交易匿名性得分较高。不过，事情具有两面性：监管与欺诈就是一场对抗游戏。对于黑客来说，虽然重用账户降低了被识别为高风险的概率，但同时，重用账户并从交易所提币也增加了黑客暴露、资金被冻结的风险。
% 716：154
% 
% 777: 94  exchange
% 710：66  exchange
% 672：21  exchange
% 689：9   exchange
% 777和710直接相连

%%%%%%%%%%%%%%%%%%%%%%%%%%%%%%%%%%%%%%%%%%%%%%%%%%%%%%%%%%%%%%%%%%%%%%%%%%%%%

%\subsection{Case Studies}

%\vspace{-2ex}
\section{Related Work}

%In this section, we present a review of recent literature on risk control studies and graph-based fraud rating techniques.

\textbf{Risk control studies in cryptocurrency.}
% 类似的研究问题上
%qs comment:加了冠词 for the Ethereum’s account model
In recent years, there has been growing interest in account clustering and detecting illicit activities (e.g., financial scams, money laundering) in cryptocurrency transaction networks~\cite{Wu2021JNCA}.
Victor~\cite{victor2020address} is the first to propose clustering heuristics for the Ethereum’s account model, including deposit address reuse, airdrop multi-participation, and self-authorization.
A recent review of the literature on cryptocurrency scams~\cite{9591634} showed that the existing methods (e.g., \cite{Ponzi2018Chen}, \cite{Farrugia2020}, and \cite{9491653}) are mainly based on supervised classifiers fed with handcrafted features. Many attempts have been made \cite{edgeprop2019Tam,elliptic2019Weber,trans2vec2020Wu} to incorporate structural information by learning the latent representations of accounts. Some researchers have investigated and modeled the money flow from a network perspective~\cite{modeling2020lin, lal2021understanding} %for
to better identify illicit activities.
% qs comment：regarding the estimation 加了冠词
After all, there is still a black area regarding the estimation of the risk value of Ethereum accounts, which is the key task in alerting about suspicious accounts and transactions on the chain.

\textbf{Rating and ranking on graph data.}
% 使用的类似的算法上
% qs comment:the authors highlighted 加了s
The aim of ratings and rankings on graph data is to provide a score or an order for each node in a graph. Currently, the main solutions are based on link analysis technique~\cite{FC2019Sangers}, Bayesian model~\cite{hooi2016birdnest}, and iterative learning~\cite{Hooi2016FRAUDAR}, etc. Similarly to the proposed \textit{RiskProp} algorithm, \cite{Kumar2018rev2} proposed axioms and iterative formulations to establish the relationship between ratings. In \cite{Bias2011Mishra}, the authors measured the bias and prestige of nodes in networks based on trust scores. In \cite{Kurshan2020}, the authors highlighted that graph-based approaches provide unique solution opportunities for financial crime and fraud detection.
A review on this topic~\cite{Akoglu2015survey} described the problems in current studies: lack of ground truths, imbalanced class, and large-scale network. These challenges also exist in our risk rating problem on Ethereum transaction networks.

%\vspace{-2ex}
\section{Conclusions and Future Work}
%qs comment: to assess the account risk加了冠词the
In this paper, we present the first systematic study to assess the account risk via a rating system named \textit{RiskProp}. In \textit{RiskProp}, we modeled transaction records of Ethereum as a bipartite graph, proposed a novel metric called de-anonymous score to quantify the transaction risk, and designed a network propagation mechanism based on transaction semantics. 
By analyzing the rating results and manually checking the accounts with high risk, we evaluated the performance of \textit{RiskProp} and obtained new insights about transaction risks on Ethereum. In addition, we employed the obtained risk scores to conduct illicit/licit account classification experiments on labeled data, and the superiority of this method over baseline methods further verified the effectiveness of \textit{RiskProp} in risk estimation. 
%In this paper, we study the risk control on Ethereum from a new perspective, called account risk rating. Based on transaction records, we calculate the de-anonymous Score and describe the network propagation mechanism. 
%The experiment results show our model can effectively estimate the risk values of accounts and catch abnormal accounts when fewer transactions occur. We also observe that risky accounts have different patterns, some of which have gang accounts. 
For future work, we plan to integrate the transaction amounts and temporal information in our model, develop a web page or online tool for querying risk values of accounts, and share the details of risky cases with the Ethereum community. 

% \section*{Acknowledgment}

% The preferred spelling of the word ``acknowledgment'' in America is without 
% an ``e'' after the ``g''. Avoid the stilted expression ``one of us (R. B. 
% G.) thanks $\ldots$''. Instead, try ``R. B. G. thanks$\ldots$''. Put sponsor 
% acknowledgments in the unnumbered footnote on the first page.

% \subsubsection*{\textbf{ANOVA Metrics:}}

% In statistics, ''MS'' shorts for Mean Square. MS between (or within) groups reflects the difference between (or within) the groups, reducing the impact of the difference in sample numbers. 
% ''F'' is the F-test statistic value. A larger F means a large variance between groups, indicating that different groups of samples are the main sources of variance. The $p$-value is the probability at the corresponding F value.
% It is generally believed that there is a statistically significant difference between the two distributions when $p$-value$<0.05$.

%\newpage

%%
%% The next two lines define the bibliography style to be used, and
%% the bibliography file.
\bibliographystyle{ACM-Reference-Format}
\bibliography{RiskProp_arxiv_submit}

\newpage

\appendix
\section{Appendix } % Only 2 pages!!!!!
\label{sec:proof}

%----------------------------------------------------------------------------------
% 字体大小：

% 七号 　　5.25pt 　　 1.845mm　　　　\tiny
% 六号 　　7.875pt 　　 2.768mm　　　　\scriptsize
% 小五号 　9pt 　　　　 3.163mm　　　　\footnotesize
% 五号 　　10.5pt 　　 3.69mm　　　　 \small
% 小四号 　12pt 　　　　4.2175mm　　　 \normalsize
% 四号 　　13.75pt 　　 4.83mm　　　　 \large
% 三号 　　15.75pt 　　 5.53mm　　　　 \Large
% 二号 　　21pt 　　　　7.38mm \LARGE
% 一号 　　27.5pt 　　 9.48mm　　　　 \huge
% 小初号 　36pt 　　　　12.65mm　　　　\Huge
% 初号 　　42pt 　　　　14.76mm
% ————————————————
% 版权声明：本文为CSDN博主「zou_albert」的原创文章，遵循CC 4.0 BY-SA版权协议，转载请附上原文出处链接及本声明。
% 原文链接：https://blog.csdn.net/zou_albert/article/details/110532165

%----------------------------------------------------------------------------------

\subsection{Proof for Prior Knowledge}
% 在这一部分，我们提供了所提出的指标满足先验知识1-5的证明。
In this part, we provide proofs that the proposed metrics, i.e., \textit{Reliability}, \textit{Trustiness}, and \textit{Confidence} satisfy Prior knowledge 1 - 5.
Prior knowledge 1 in the main paper is the following:

\noindent \textbf{[Prior knowledge 1] }
Payees with higher trustiness receive transactions with higher de-anonymous scores. Formally, if two payees $v_1$ and $v_2$ have a one-to-one mapping,  $h:In(v_1)\rightarrow In(v_2)$ and $Score(u,v_1) > Score(h(u), v_2)$ $\forall(u,v_1)\in In(v_1)$, then $T(v_1)>T(v_2)$.

The formulation to be used to show that the prior knowledge is satisfied is Equations~\ref{T}, \ref{C}, and~\ref{R} in the main paper.

\begin{footnotesize}
    \begin{equation}
        T(v) = \frac{\sum_{(u,v)\in In(v)} Score(u,v)\times Conf(u,v)}{|In(v)|} \nonumber 
    \end{equation}
    \begin{equation}
        R(u) = \frac{\sum_{(u,v)\in Out(u)} Conf(u,v)}{|Out(u)|} \nonumber 
    \end{equation}
    \begin{equation}
        Conf(u, v) = \frac{R(u) + (1-|Score(u,v)-T(v)|)}{2} \nonumber 
    \end{equation}    
\end{footnotesize}

\begin{proof}
    To prove the Prior Knowledge 1, let us take two payees $v_1$ and $v_2$ that have identically ego networks and a one-to-one mapping $h$, such that $|In(v_1)|=|In(v_2)|$, $Conf(u, v_1)=Conf(h(u), v_2)$, and $Score(u,v_1) > Score(h(u), v_2)$ $\forall(u,v_1)\in In(v_1)$. 
    
    According to Equation~\ref{T}, we have % substract $T(v_2)$ from $T(v_1)$
    
    \begin{scriptsize}
    \begin{equation}
    \begin{split}
    % |T^{\infty}(v)-T^1(v)| & = \frac{1}{|In(v)|}\times  (|\sum_{(u,v)\in In(v)}Score(u,v)dot Conf^{\infty}(u,v) \\
    % & -\sum_{(u,v)\in In(v)}Score(u,v)\times  Conf^{0}(u,v)|)\\
        T(v_1)-T(v_2)&= \frac{\sum_{(u,v_1)\in In(v_1)} Score(u,v_1)\times Conf(u,v_1)}{|In(v_1)|}-\\
        &\frac{\sum_{(u,v_2)\in In(v_2)} Score(h(u),v_2)\times Conf(h(u),v_2)}{|In(v_2)|}\\
        &=\frac{\sum_{(u,v_1)\in In(v_1)}(Score(u,v_1)-Score(h(u),v_2))\times Conf(u,v_1)}{|In(v_1)|}
    \end{split}
    \nonumber 
    \end{equation}
    \end{scriptsize}
    
    As $Score(u,v_1) > Score(h(u), v_2)$, so
    
    \begin{footnotesize}
    \begin{equation}
    \begin{split}
    T(v_1)-T(v_2) > \frac{\sum_{(u,v_1)\in In(v_1)}Conf(u,v_1)}{|In(v_1)|}
    \end{split}
    \nonumber 
    \end{equation}
    \end{footnotesize}
    
    As $Conf(u,v_1)$ $\ge$ 0 because \textit{Confidence} are non-negative,
    \begin{footnotesize}
    \begin{equation}
    \begin{split}
    T(v_1)-T(v_2) > 0 \Rightarrow T(v_1) > T(v_2) 
    \end{split}
    \nonumber 
    \end{equation}
    \end{footnotesize}
    
    The other items of prior knowledge have very similar and straightforward proof.
\end{proof}

\subsection{Proof for Convergence}

%首先给出传播方程随时间变化关系
Before the proof of convergence, we first discuss the boundary of proposed metrics. 
At the end of iteration $t$ of Algorithm~\ref{alg:riskprop}, and by equation~\ref{T}, \ref{C}, and~\ref{R}, we get,

% First\blue{,} we present the propagation equation\blue{:}
\begin{footnotesize}
    \begin{equation}
    \label{eq1}
        T^t(v)=\frac{\sum_{(u,v)\in In(v)}Score(u,v)\times  Conf^{t-1}(u,v)}{|In(v)|}
    \nonumber 
    \end{equation}
    \begin{equation}
        R^t(u)=\frac{\sum_{(u,v)\in Out(u)}Conf^{t-1}(u,v)}{|Out(u)|}
        \nonumber 
    \end{equation}
    \begin{equation}
        Conf^t(u,v) = \frac{R^t(u)+(1-|Score(u,v)-T^t(v)|}{2})
        \nonumber 
    \end{equation}
\end{footnotesize}

$T^{\infty(v)}, R^{(\infty)}(u), Conf^{(\infty)}(u,v)$ are their final values after convergence.

%方程收敛时我们可以得到:
% We can get the following equations once they converge.
% \begin{equation}
% \label{eq1}
%     T^{\infty}(v)=\frac{\sum_{(u,v)\in In(v)}Score(u,v)\times  Conf^{\infty}(u,v)}{|In(v)|}
% \end{equation}
% \begin{equation}
%     R^{(\infty)}(u)=\frac{\sum_{(u,v)\in Out(u)}Conf^{\infty}(u,v)}{|Out(u)|}
% \end{equation}
% \begin{equation}
%     Conf^{(\infty)}(u,v) = \frac{R^{\infty}(u)+(1-{|Score(u,v)-T^{\infty}(v)|}}{2})
% \end{equation}

\begin{lemma}(Boundary discussion)
    %迭代有界性:设网络之中交易去匿名分数最大值为S,可信度,可靠性以及conf计算如上述.则
    Set the maximum score in the transaction network as $M$, namely:
    \begin{small}
    \begin{equation}
        M = \mathop{\max}_{(u,v)} Score(u,v)
        \nonumber 
    \end{equation}
    \end{small}
    
    then $|M| < 1$.
    
    % The value of \textit{Trustiness}, \textit{Reliability}, and \textit{Confidence} are bounded as follows:    
    The difference between a payee $v$’s final \textit{Trustiness} and its \textit{Trustiness} after the first iteration is

    \begin{footnotesize}
        \begin{equation}
            |T^{\infty}(u)-T^1(u)|\le |M|
            \nonumber 
        \end{equation}
    \end{footnotesize}
    
    Similarly,
    
    \begin{footnotesize}
        \begin{equation}
        |R^{\infty}(u)-R^1(u)|\le 1
        \nonumber 
        \end{equation}
        \begin{equation}
        |Conf^{\infty}(u,v)-Conf^1(u,v)|\le \frac{1+|M|}{2} = \alpha~(\alpha \le 1)
        \nonumber 
        \end{equation}
    \end{footnotesize}
    
\end{lemma}

\begin{proof}
%\renewcommand{\qedsymbol}{}
% Since $Score$ falls at range \blue{$[-1,1]$}, we need to discuss the case when $|S|=1$:
%因为取值范围在【0，1】所以...
%简单带过，压缩
% 首先，我们说明我们说明S是严格小于1的，在实际使用中。从正文的Score的公式中，我们可以知道，让一个交易的Score=1，意味着
First, we state that $|M|$ is strictly less than 1 in practice. 
According to the formulation of $Score$ of the main paper, we can see that $Score(u,v) = 1$ when $OutTxn(u)=maxOut$ and $InTxn(v)=maxIn$, it is an extreme situation where the largest number of payments and receptions of the entire network appears in one transaction. The other case is $Score(u,v) = -1$ when that $OutTxn(u)=1$ and $InTxn(v)=1$ at the same time. This situation presents to be some isolated transactions, however, they do not propagate risk and thus do not influence convergence. These situations are out of our consideration. So we get $|M|$ is strictly smaller than 1.
%Generally, $|S|$ is strictly smaller than 1. 
%\textbf{证明：}首先我们计算账户的Trust(v)的迭代有界性，
    
%qs comment:这里不算是First of all的感觉。最初证明了|S|小于1
%First of all, 
Then, we prove that $T(v)$ is bounded during the iterations:

\begin{footnotesize}
    \begin{equation}
    \begin{split}
    % |T^{\infty}(v)-T^1(v)| & = \frac{1}{|In(v)|}\times  \left(|\sum_{(u,v)\in In(v)}Score(u,v)\times  Conf^{\infty}(u,v) \\
    % & -\sum_{(u,v)\in In(v)}Score(u,v)\times  Conf^{0}(u,v)|\right)
    |T^{\infty}(v)-T^1(v)| & = | \frac{\sum_{(u,v)\in In(v)}Score(u,v)\times Conf^{\infty}(u,v)}{|In(v)|}-\\
    &\frac{\sum_{(u,v)\in In(v)}Score(u,v)\times  Conf^{0}(u,v)}{|In(v)|}|
    \end{split}
    \nonumber 
    \end{equation}
\end{footnotesize}

Since $|x+y|\le |x|+|y|$, we get,

\begin{scriptsize}
\begin{equation}
\begin{split}
    |T^{\infty}(v)-T^1(v)| &\le \frac{\sum_{(u,v)\in In(v)}|Score(u,v)\times (Conf^{\infty}(u,v)  - Conf^{0}(u,v))|}{|In(v)|}
% |T^{\infty}(v1(v)-T^1(v)| &\le \frac{1}{|In(v)|}\times  (\sum_{(u,v)\in In(v)}|Score(u,v)\times (Conf^{\infty}(u,v)  - Conf^{0}(u,v))|)
\end{split}
\nonumber 
\end{equation}
\end{scriptsize}
Since $|x\times  y| = |x|\times |y|$, we have,
\begin{scriptsize}
\begin{equation}
\begin{split}
% |T^{\infty}(v)-T^1(v)|  &\le \frac{1}{|In(v)|}\times  (\sum_{(u,v)\in In(v)}|Score(u,v)|\times  (Conf^{\infty}(u,v) - Conf^{0}(u,v))|)
|T^{\infty}(v)-T^1(v)|  &\le \frac{\sum_{(u,v)\in In(v)}|Score(u,v)|\times  (Conf^{\infty}(u,v) - Conf^{0}(u,v))|}{|In(v)|}
\end{split}
\label{T_1}
\end{equation}
\end{scriptsize}
%由于我们已知$|Score(u,v)| \le 1$并且$|(Confidence^{\infty}(u,v) - Confidence^{0}(u,v))| \le 1$，所以可以得到
Since $|Score(u, v)| \le |M| \le 1 $, and $|(Conf^{\infty}(u,v) - Conf^{0}(u,v))| \le 1$, we get,

\begin{scriptsize}
\begin{equation}
\begin{split}
|T^{\infty}(v)-T^1(v)|  \le |M| \times  \frac{|In(v)|}{|In(v)|} = |M|
\end{split}
\nonumber 
\end{equation}    
\end{scriptsize}

%接下来我们计算账户的Reliable(u)的迭代有界性，
Next, we conduct the proof on $R(u)$:
\begin{scriptsize}
\begin{equation}
\begin{split}
% |R^{\infty}(u)-R^1(u)| & = \frac{1}{|Out(u)|}\times  (|\sum_{(u,v)\in Out(u)} Conf^{\infty}(u,v)  -\sum_{(u,v)\in In(v)}Conf^{0}(u,v)|)
|R^{\infty}(u)-R^1(u)| & = \frac{|\sum_{(u,v)\in Out(u)} Conf^{\infty}(u,v)  -\sum_{(u,v)\in Out(u)}Conf^{0}(u,v)|}{|Out(u)|}
\end{split}
\nonumber 
\end{equation}
\end{scriptsize}

%由于$|x+y|\le |x|+|y|$，我们可以得到
Again, since $|x\times  y| = |x|\times |y|$, we get,

\begin{scriptsize}
\begin{equation}
\begin{split}
% &|R^{\infty}(u)-R^1(u)| \le \frac{1}{|Out(u)|}\times  (\sum_{(u,v)\in Out(u)}| Conf^{\infty}(u,v) - Conf^{0}(u,v)|)
&|R^{\infty}(u)-R^1(u)| \le \frac{\sum_{(u,v)\in Out(u)}| Conf^{\infty}(u,v) - Conf^{0}(u,v)|}{|Out(u)|}
\end{split}
\label{R_1}
\end{equation}
    
\end{scriptsize}

%由于我们已知$|(Confidence^{\infty}(u,v) - Confidence^{0}(u,v))| \le 1$，所以可以得到
Similarly, since $|(Conf^{\infty}(u,v) - Conf^{0}(u,v))| \le 1$, we have,
\begin{scriptsize}
 \begin{equation}
\begin{split}
|R^{\infty}(u)-R^1(u)|  \le \frac{|Out(u)|}{|Out(u)|} = 1
\end{split}
\nonumber 
\end{equation}   
\end{scriptsize}

%最后我们要计算交易的Confidence(u,v)的迭代有界性，
%-qscomment-----------||有问题，需要调整，明早调
%----------修改后：
Finally, we calculate the bound of  $Conf(u,v) $:
\begin{scriptsize}
\begin{equation}
\begin{split}
% &|Conf^{\infty}(u,v)-Conf^1(u,v)|  = \frac{1}{2}\times  (|R^{\infty}(u)-R^1(u)+ \\
% &(|Score(u,v)-T^1(v)|-|Score(u,v)-T^{\infty}(v)|)|)
&|Conf^{\infty}(u,v)-Conf^1(u,v)|  =\\ &\frac{|R^{\infty}(u)-R^1(u)+|Score(u,v)-T^1(v)|-|Score(u,v)-T^{\infty}(v)||}{2}
\end{split}
\nonumber 
\end{equation}
\end{scriptsize}

Since $|x+y|\le |x|+|y|$, we have

\begin{scriptsize}
\begin{equation}
\begin{split}
    % \Rightarrow
%     &|Conf^{\infty}(u,v)-Conf^1(u,v)| \le \frac{1}{2}\times  (|R^{\infty}(u)-R^1(u)|+ \\
%  &(||Score(u,v)-T^1(v)|-|Score(u,v)-T^{\infty}(v)||))
 &|Conf^{\infty}(u,v)-Conf^1(u,v)| \le\\
 &\frac{|R^{\infty}(u)-R^1(u)|+
 (||Score(u,v)-T^1(v)|-|Score(u,v)-T^{\infty}(v)||)}{2}
%     &|Conf^{\infty}(u,v)-Conf^1(u,v)| \le \frac{1}{2}\times  (|R^{\infty}(u)-R^1(u)|+ \\
%  &(||Score(u,v)-T^1(v)|-|Score(u,v)-T^{\infty}(v)||))
\end{split}
\nonumber 
\end{equation}
\end{scriptsize}

% Finally, we calculate the bound of  $Conf(u,v) $:
% \begin{scriptsize}
% \begin{equation}
% \begin{split}
% &|Conf^{\infty}(u,v)-Conf^1(u,v)|  = \frac{1}{2}\times  (|R^{\infty}(u)-R^1(u)|+ \\
% &(||Score(u,v)-T^{\infty}(v)|-|Score(u,v)-T^1(v)||))
% \end{split}
% \end{equation}
% \end{scriptsize}

%由于$||x|-|y|| \le |x-y|$,所以可以得到，
Since $||x|-|y|| \le |x-y|$, it follows that,
\begin{scriptsize}
\begin{equation}
\begin{split}
% &|Conf^{\infty}(u,v)-Conf^1(u,v)| \\
% &\le \frac{1}{2}\times  (|R^{\infty}(u)-R^1(u)|+ |T^{\infty}(v)-T^1(v)|)
&|Conf^{\infty}(u,v)-Conf^1(u,v)| \le \frac{|R^{\infty}(u)-R^1(u)|+ |T^{\infty}(v)-T^1(v)|}{2}
\end{split}
\label{C_1}
\end{equation}    
\end{scriptsize}

%由于我们已知$|Reliable^{\infty}(u) - Reliable^{1}(u)| \le 1$，并且$|Trust^{\infty}(v)-Trust^1(v)| \le 1$,所以可以得到
Since $|R^{\infty}(u) - R^{1}(u)| \le 1$, and$|T^{\infty}(v)-T^1(v)| \le |M|$, we get,

\begin{scriptsize}
\begin{equation}
\begin{split}
|Conf^{\infty}(u,v)-Conf^1(u,v)|  \le  \frac{1+|M|}{2}
\end{split}
\nonumber 
\end{equation}    
\end{scriptsize}

For convenience, we let $\frac{1+|M|}{2}=\alpha$. Since $|M|<1$, then $\alpha<1$.

%由以上得证模型的迭代有界性。
\end{proof}

\begin{theorem}
    Convergence of Propagation: %Confidence of a transaction converges to a fixed value $Conf^{\infty}(u,v)$ as iterations increase, similarly Trust and Reliable also converge to a fixed value. Namely:
    The difference during iterations is bounded as as $|Conf^{\infty}(u,v)-Conf^t(u,v)| \le \alpha^{t}$ ($\alpha=\frac{1+|M|}{2} < 1$), $\forall(u, v)\in S$. As $t$ increases, the difference decreases and $Conf^t(u,v)$ converges to $|Conf^{\infty}(u,v)$. Similarly, $|T^{\infty}(v)-T^t(v)| \le \alpha^{t-1},\forall v \in V$, $|R^{\infty}(u)-R^t(u)| \le \alpha^{t-1},\forall u \in U$.
    
    % 这几个写出来好像意义不大，是显而易见的？
    % \begin{scriptsize}
    %  \begin{align}
    %      & \lim_{t \to \infty} Conf^{t}(u,v) = Conf^{\infty}(u,v)\nonumber \\
    %      & \lim_{t \to \infty} R^{t}(u,v) = R^{\infty}(u,v)\nonumber \\
    %      & \lim_{t \to \infty} T^{t}(u,v) =T^{\infty}(u,v)
    %      \nonumber 
    % \end{align}   
    % \end{scriptsize}

    %交易的Confidence随迭代过程存在着上边界$|Confidence^{\infty}(u,v)-Confidence^t(u,v)| \le (\frac{3}{4})^{t},\forall (u,v)$，随着迭代次数增加，最后Confidence会收敛到一个固定值$Confidence^{\infty}(u,v)$，类似的账户的Trust和Reliable也会收敛到一个固定值。
\end{theorem}
\begin{proof}

    Similar to Equations~\ref{T_1}, \ref{R_1}, and \ref{C_1}, we have,
%-----------要编号---------
\begin{scriptsize}
    \begin{equation}
    \label{equ:Conf}
    \begin{split}
        % &|Conf^{\infty}(u,v)-Conf^t(u,v)| \le \frac{1}{2}\times  (|R^{\infty}(u)-R^t(u)|+ |T^{\infty}(v)-T^t(v)|)
        &|Conf^{\infty}(u,v)-Conf^t(u,v)| \le \frac{|R^{\infty}(u)-R^t(u)|+ |T^{\infty}(v)-T^t(v)|}{2}
    \end{split}
    \end{equation}    
    %-----------要编号---------
    %由Reliable关于有界性的证明可知，
    % Similarly, we get the following formulas from the proof of \textit{Reliability} is bounded:
    \begin{equation}
    \label{equ:R}
    \begin{split}
    % &|R^{\infty}(u)-R^t(u)|\le \frac{1}{|Out(u)|}\times  (\sum_{(u,v,)\in Out(u)}| Conf^{\infty}(u,v) - Conf^{t-1}(u,v)|)
    &|R^{\infty}(u)-R^t(u)|\le \frac{\sum_{(u,v,)\in Out(u)}| Conf^{\infty}(u,v) - Conf^{t-1}(u,v)|}{|Out(u)|}
    \end{split}
    \end{equation}
    %-----------要编号---------
    %由Trust关于有界性的证明可知，
    %According to the proof of \textit{Trustiness} is bounded:
    \begin{equation}
    \label{equ:T}
    \begin{split}
        % &|T^{\infty}(v)-T^t(v)| \\
        % &\le \frac{1}{|In(v)|}\times  (\sum_{(u,v)\in In(v)}|Score(u,v)|\times  |(Conf^{\infty}(u,v)- Conf^{t-1}(u,v))|)
        &|T^{\infty}(v)-T^t(v)| \\
        &\le \frac{\sum_{(u,v)\in In(v)}|Score(u,v)|\times  |(Conf^{\infty}(u,v)- Conf^{t-1}(u,v))|}{|In(v)|}
    \end{split}
    \end{equation}
\end{scriptsize}

%首先证明交易的Confidence的收敛性，由Confidence关于有界性的证明可知,
% First, we prove the convergence of \textit{Confidence}. %We obtain the middle result from our proof of the bound of \textit{Confidence} that 
% %运用数学归纳法，当$t=1$时，$|Confidence^{\infty}(u,v)-Confidence^1(u,v)|\le \frac{3}{4}$成立，假设$|Confidence^{\infty}(u,v)-Confidence^{t-1}(u,v)|\le (\frac{3}{4})^{t-1}$成立,则有
% The theorem is proven by induction:
First, we will prove the convergence of \textit{Confidence} using mathematical induction.

\textbf{Base case of induction.}

When $t=1$, as we proved in Lemma A.1, we get:
\begin{footnotesize}
    \begin{equation}
    |Conf^{\infty}(u,v)-Conf^1(u,v)|\le \alpha^{1}
    \nonumber 
    \end{equation}
\end{footnotesize}

\textbf{Induction step.}  

We assume by hypothesis that

\begin{footnotesize}
\begin{equation}
    |Conf^{\infty}(u,v)-Conf^{t-1}(u,v)|\le \alpha^{t-1},
    \nonumber 
\end{equation} 
\end{footnotesize}

which is consistent with the base case already.

Then, by substituting Equations~\ref{equ:R} and~\ref{equ:T} into Equation~\ref{equ:Conf}, for the case in the next iteration where time is $t$, we have,
\begin{scriptsize}
    \begin{equation}
    \begin{split}
    % &|Conf^{\infty}(u,v)-Conf^t(u,v)|  \\
    % &\le \frac{1}{2}\times  (|\frac{1}{|Out(u)|}\times  \left(\sum_{(u,v)\in Out(u)}| Conf^{\infty}(u,v) - Conf^{t-1}(u,v)|\right)| \\
    % &+\frac{1}{|In(v)|}\times  (\sum_{(u,v)\in In(v)}|Score(u,v)|\times |(Conf^{\infty}(u,v) - Conf^{t-1}(u,v))|)) \\
    % &\le \frac{1}{2}\times  ((\frac{1+|S|}{2})^{t-1} + \frac{|S|}{|In(v)|}\times (\sum_{(u,v)\in In(v)} |(Conf^{\infty}(u,v) - Conf^{t-1}(u,v))|)) \\
    % &\le \frac{1}{2}\times  ((\frac{1+|S|}{2})^{t-1} +|S|  \times  (\frac{1+|S|}{2})^{t-1}) \\
    % &\le  (\frac{1+|S|}{2})^{t} 
    &|Conf^{\infty}(u,v)-Conf^t(u,v)|  \\
    &\le \frac{\sum_{(u,v)\in Out(u)}| Conf^{\infty}(u,v) - Conf^{t-1}(u,v)|}{2\times|Out(u)|})\\
    &+\frac{\sum_{(u,v)\in In(v)}|Score(u,v)|\times |(Conf^{\infty}(u,v) - Conf^{t-1}(u,v)|}{2\times|In(v)|}\\
    &\le \frac{1}{2}\times   \left((\frac{1+|M|}{2})^{t-1} + \frac{|M|\times \sum_{(u,v)\in In(v)} |Conf^{\infty}(u,v) - Conf^{t-1}(u,v)|}{|In(v)|}\right)\\
    &\le \frac{1}{2}\times  \left((\frac{1+|M|}{2})^{t-1} +|M|  \times  (\frac{1+|M|}{2})^{t-1}\right) \\
    &\le  (\frac{1+|M|}{2})^{t} = \alpha^{t} \\
    \end{split}
    \nonumber 
    \end{equation}    
\end{scriptsize}

Therefore, $|Conf^{\infty}(u,v)-Conf^t(u,v)|\le \alpha^{t}$.

%%%%%%%%%代入公式(26)中我们可以得到，
\begin{scriptsize}
\begin{equation}
\begin{split}
% |R^{\infty}(u)-R^t(u)| & \le \frac{1}{|Out(u)|}\times  (\sum_{(u,v)\in Out(u)}| Conf^{\infty}(u,v) - Conf^{t-1}(u,v)|) \\
% &\le \frac{1}{|Out(u)|}\times  (\sum_{(u,v)\in Out(u)} (\frac{1+|S|}{2})^{t-1}) \\
% &\le (\frac{1+|S|}{2})^{t-1}
|R^{\infty}(u)-R^t(u)| & \le \frac{\sum_{(u,v)\in Out(u)}| Conf^{\infty}(u,v) - Conf^{t-1}(u,v)|}{|Out(u)|} \\
&\le \frac{\sum_{(u,v)\in Out(u)} (\frac{1+|M|}{2})^{t-1}}{|Out(u)|}\\
&\le (\frac{1+|M|}{2})^{t-1} = \alpha^{t-1}
\end{split}
\nonumber 
\end{equation}
  
\end{scriptsize}

%%%%%%%%%%%代入公式(27)中我们可以得到，
\begin{scriptsize}

\begin{equation}
\begin{split}
% &|T^{\infty}(v)-T(v)^t|  \\
% &\le \frac{1}{|In(v)|}\times  (\sum_{(u,v)\in In(v)}|Score(u,v)|\times  |(Conf^{\infty}(u,v) - Conf^{t-1}(u,v))|) \\
% &\le \frac{1}{|In(v)|}\times  (\sum_{(u,v)\in In(v)}|(Conf^{\infty}(u,v) - Conf^{t-1}(u,v))|) \\
% &\le \frac{1}{|In(v)|}\times  (\sum_{(u,v)\in In(v)}(\frac{1+|S|}{2})^{t-1}) \\
% &\le (\frac{1+|S|}{2})^{t-1} 
&|T^{\infty}(v)-T(v)^t|  \\
&\le \frac{\sum_{(u,v)\in In(v)}|Score(u,v)|\times  |(Conf^{\infty}(u,v) - Conf^{t-1}(u,v))|}{|In(v)|} \\
&\le \frac{\sum_{(u,v)\in In(v)}|(Conf^{\infty}(u,v) - Conf^{t-1}(u,v))|}{|In(v)|}\\
&\le \frac{\sum_{(u,v)\in In(v)}(\frac{1+|M|}{2})^{t-1}}{|In(v)|} \\
&\le (\frac{1+|M|}{2})^{t-1} = \alpha^{t-1}
\end{split}
\nonumber 
\end{equation}
\end{scriptsize}

As discussed in the Lemma A.1. we know that $|M|$ is strictly smaller than 1, then we have $\alpha <1$. As $t$ increases, $\alpha^{t-1} \rightarrow 0$ and $\alpha^{t} \rightarrow 0$, so after \textit{t} iterations, $Conf(u,v)^t\rightarrow Conf^{\infty}(u,v)$, $R(u)^t\rightarrow R^{\infty}(u)$, and $T(v)^t\rightarrow T^{\infty}(v)$, the algorithm converges.
% \rightarrow 0, (\frac{1+|M|}{2})^{t} \rightarrow 0, so after \textit{t} iterates, T(v)^t\rightarrow T^{\infty}(v)}
% which ensures the convergence of \textit{Confidence}, \textit{Reliability}, and \textit{Trustiness}.

%由以上证明可知，当迭代次数t趋于无穷时，账户的Trust,Reliable以及交易的Confidence会收敛到一个固定值。
\end{proof}

\subsection{Proof for Uniqueness}
In this part, we provides proofs that \textit{Reliability}, \textit{Trustiness}, and \textit{Confidence} are unique. 

\begin{theorem}
    Confidence, Reliability, and Trustiness converge to the unique value. 
\end{theorem}

\begin{proof}
%\renewcommand{\qedsymbol}{}
%首先我们考虑交易Confidence的唯一性，我们假设$Confidence(u,v)$会收敛到两个不同的值，并且之间的差值为D，并且由收敛性证明中的(25),(26),(27)可知，
First, we consider the uniqueness of \textit{Confidence} using mathematical contradiction. % 反证法

Let the $Conf(u,v)$ converges to different values. So, let $(u,v)$ be the transaction with maximum \textit{Confidence} difference, $D$ (with $D \ge 0$), between its two possible $Conf_1(u,v)$ and $Conf_2(u,v)$.

According to Equation \ref{equ:Conf}, we get,

\begin{footnotesize}
\begin{equation}
\label{equ:U1}
\begin{split}
D&=|Conf_1^{\infty}(u,v)-Conf_2^{\infty}(u,v)| \\ 
&\le\frac{|R_1^{\infty}(u)-R_2^{\infty}(u)|+|T_1^{\infty}(v)-T_2^{\infty}(v)|}{2}
\end{split}
\end{equation}  
\end{footnotesize}

Then, according to Equation \ref{equ:R} and \ref{equ:T}, we have,

\begin{scriptsize}
\begin{equation}
\label{equ:U2}
\begin{split}
% |R_1^{\infty}(u)-R_2^{\infty}(u)| 
% \le \frac{1}{|Out(u)|}\times  (\sum_{(u,v)\in Out(u)}|
% Conf_1^{\infty}(u,v) - Conf_2^{\infty}(u,v)|)
% \le D
|R_1^{\infty}(u)-R_2^{\infty}(u)| & \le \frac{\sum_{(u,v)\in Out(u)}| Conf_1^{\infty}(u,v) - Conf_2^{\infty}(u,v)|}{|Out(u)|} \\
& \le D
\end{split}
\end{equation} 

\begin{equation}
\label{equ:U3}
\begin{split}
|T_1^{\infty}(v)-T_2^{\infty}(v)| &\le \frac{\sum_{(u,v)\in In(v)}|Score(u,v)|\times  |Conf_1^{\infty}(u,v)- Conf_2^{\infty}(u,v)|}{|In(v)|} \\
& \le |M|\times D
\end{split}
\end{equation} 
\end{scriptsize}

%将公式(32).(33)代入公式(31)可得，

We substitute Equation \ref{equ:U2} and \ref{equ:U3} into Equation (\ref{equ:U1}), and get,
\begin{footnotesize}
\begin{equation}
\begin{split}
&D=|Conf_1^{\infty}(u,v)-Conf_2^{\infty}(u,v)|  \\
&\le \frac{1}{2}\times  (\frac{\sum_{(u,v)\in Out(u)}| Conf_1^{\infty}(u,v) - Conf_2^{\infty}(u,v)|}{|Out(u)|} \\
&+\frac{\sum_{(u,v)\in In(v)}|Score(u,v)|\times |Conf_1^{\infty}(u,v) - Conf_2^{\infty}(u,v)|}{|In(v)|} \\
&\le \frac{1}{2}\times  \left( D +\frac{|M|\times\sum_{(u,v)\in In(v)} |Conf^{\infty}(u,v) - Conf^{\infty}(u,v)|}{|In(v)|} \right)\\
&\le \frac{1}{2}\times ( D +|M| \times  D ) \\
&\le  (\frac{1+|M|}{2}) \times  D \\
& = \alpha \times  D \\
% &D=|Conf_1^{\infty}(u,v)-Conf_2^{\infty}(u,v)|  \\
% &\le \frac{1}{2}\times  (\frac{1}{|Out(u)|}\times  (\sum_{(u,v)\in Out(u)}| Conf_1^{\infty}(u,v) - Conf_2^{\infty}(u,v)|)| \\
% &+\frac{1}{|In(v)|}\times  (\sum_{(u,v)\in In(v)}|Score(u,v)|\times  |(Conf_1^{\infty}(u,v) - Conf_2^{\infty}(u,v))|)) \\
% &\le \frac{1}{2}\times  ( D +\frac{|S|}{|In(v)|}\times (\sum_{(u,v)\in In(v)} |(Confidence^{\infty}(u,v) - Confidence^{\infty}(u,v))|)) \\
% &\le \frac{1}{2}\times  ( D +|S| \times  D) \\
% &\le  \frac{1+|S|}{2} \times  D \\
% &\le D
\end{split}
\nonumber 
\end{equation}
\end{footnotesize}

Thus, by solving $D \le\alpha \times  D (\alpha \ne 0)$ and with the condition that $D\ge 0$, we obtain $D=0$.
Then, $|Conf_1^{\infty}(u,v)-Conf_2^{\infty}(u,v)|=0$ and converge value of \textit{Confidence} is unique. The uniqueness of \textit{Trustiness} and \textit{Reliability} have similar proof.
%所以可知$D=0$，唯一性成立，而迭代中交易的Confidence满足唯一性，相应的账户的Trust以及Reliable也满足唯一性，唯一性得证。

\end{proof}
\end{document}